\documentclass{llncs}

\usepackage{proof,latexsym, amsmath, amsfonts,amssymb}

\newcommand{\texcomment}[1]{}

\texcomment{
\newtheorem{theorem}{Theorem}[section]
\newtheorem{definition}[theorem]{Definition}
\newtheorem{lemma}[theorem]{Lemma}

\newenvironment{proof}{\noindent\rm{\bf Proof:}}{\hbox{$\Box$}\vspace*{0.2\baselineskip}}

}

\newenvironment{reflemma}[1]{\begin{trivlist}\item[\hskip
      \labelsep{\bf Lemma #1.}]\it}{\end{trivlist}}

\newenvironment{reftheorem}[1]{\begin{trivlist}\item[\hskip
      \labelsep{\bf Theorem #1.}]\it}{\end{trivlist}}

\newcommand{\aset}[1]{\{{#1}\}}
\newcommand{\aseq}[1]{\langle#1\rangle}

\newcommand{\sembrack}[1]{[\hspace*{-1.2pt}[#1]\hspace*{-1.2pt}]}
\newcommand{\dom}{\mathop\textit{dom}\nolimits}
\newcommand{\ran}{\mathop\textit{ran}\nolimits}

\newcommand{\ttif}[3]{\textsf{if}\;{#1}\;\textsf{then}\;{#2}\;\textsf{else}\;{#3}}
\newcommand{\ttassign}[2]{{#1}\;:=\;{#2}}

\newcommand{\wpre}[2]{\textit{wp}({#1}, {#2})}

\newcommand{\vect}[1]{\overrightarrow{{#1}}}

\usepackage{setspace}
\begin{document}

\title{On Bounding Problems of Quantitative Information
  Flow\thanks{This work was supported by MEXT KAKENHI 20700019,
    20240001, and 22300005, and Global COE Program ``CERIES.''}}
\author{Hirotoshi Yasuoka\inst{1} \and Tachio Terauchi\inst{2}}
\institute{Tohoku University\\ \email{yasuoka@kb.ecei.tohoku.ac.jp}
  \and Nagoya University\\\email{terauchi@is.nagoya-u.ac.jp}}
\maketitle

\begin{abstract}
Researchers have proposed formal definitions of quantitative
information flow based on information theoretic notions such as the
Shannon entropy, the min entropy, the guessing entropy, belief, and
channel capacity.  This paper investigates the hardness of precisely
checking the quantitative information flow of a program according to
such definitions.  More precisely, we study the ``bounding problem''
of quantitative information flow, defined as follows: Given a program
$M$ and a positive real number $q$, decide if the quantitative
information flow of $M$ is less than or equal to $q$.  We prove that
the bounding problem is not a $k$-safety property for any $k$ (even
when $q$ is fixed, for the Shannon-entropy-based definition with the
uniform distribution), and therefore is not amenable to the
self-composition technique that has been successfully applied to
checking non-interference.  We also prove complexity theoretic
hardness results for the case when the program is restricted to
loop-free boolean programs.  Specifically, we show that the problem is
PP-hard for all definitions, showing a gap with non-interference which
is coNP-complete for the same class of programs.  The paper also
compares the results with the recently proved results on the
comparison problems of quantitative information flow.\newline

\end{abstract}

\section{Introduction}

\label{sec:introduction}

We consider programs containing high security inputs and low security
outputs.  Informally, the quantitative information flow problem
concerns the amount of information that an attacker can learn about
the high security input by executing the program and observing the low
security output.  The problem is motivated by applications in
information security.  We refer to the classic by
Denning~\cite{denning82} for an overview.

In essence, quantitative information flow measures {\em how} secure,
or insecure, a program (or a part of a program --e.g., a variable--)
is.  Thus, unlike
non-interference~\cite{DBLP:conf/sosp/Cohen77,goguen:sp1982}, that
only tells whether a program is completely secure or not completely
secure, a definition of quantitative information flow must be able to
distinguish two programs that are both interferent but have different
degrees of ``secureness.''

For example, consider the following programs.
\[
\begin{array}{l}
M_1 \equiv \ttif{H = g}{\ttassign{O}{0}}{\ttassign{O}{1}} \\
M_2 \equiv \ttassign{O}{H}
\end{array}
\]
In both programs, $H$ is a high security input and $O$ is a low
security output.  Viewing $H$ as a password, $M_1$ is a prototypical
login program that checks if the guess $g$ matches the
password.\footnote{Here, for simplicity, we assume that $g$ is a
  program constant.  See Section~\ref{sec:prelim} for modeling
  attacker/user (i.e., low security) inputs.}  By executing $M_1$, an
attacker only learns whether $H$ is equal to $g$, whereas she would be
able to learn the entire content of $H$ by executing $M_2$.  Hence, a
reasonable definition of quantitative information flow should assign a
higher quantity to $M_2$ than to $M_1$, whereas non-interference would
merely say that $M_1$ and $M_2$ are both interferent, assuming that
there are more than one possible values of $H$.

Researchers have attempted to formalize the definition of quantitative
information flow by appealing to information theory.  This has
resulted in definitions based on the Shannon
entropy~\cite{denning82,clarkjcs2007,malacaria:popl2007}, the min
entropy~\cite{smith09}, the guessing
entropy~\cite{kopf07,DBLP:conf/sp/BackesKR09},
belief~\cite{clarkson:csf2005}, and channel
capacity~\cite{mccamant:pldi2008,malacaria08,NMS2009}.  All of these
definitions map a program (or a part of a program) onto a non-negative
real number, that is, they define a function $\mathcal{X}$ such that
given a program $M$, $\mathcal{X}(M)$ is a non-negative real number.
(Concretely, $\mathcal{X}$ is ${\it SE}[\mu]$ for the
Shannon-entropy-based definition with the distribution $\mu$, ${\it
  ME}[\mu]$ for the min-entropy-based definition with the distribution
$\mu$, ${\it GE}[\mu]$ for the guessing-entropy-based definition with
the distribution $\mu$, and ${\it CC}$ for the channel-capacity-based
definition.\footnote{The belief-based definition takes additional
  parameters as inputs, and is discussed below.})  Therefore, a
natural verification problem for quantitative information flow is to
decide, given $M$ and a quantity $q \geq 0$, if $\mathcal{X}(M) \leq
q$.  The problem is well-studied for the case $q = 0$ as it is
actually equivalent to checking non-interference
(cf. Section~\ref{sec:nonint}).  The problem is open for $q > 0$ .  We
call this the {\em bounding problem} of quantitative information flow.

The problem has a practical relevance as a user is often interested in
knowing if her program leaks information within some allowed bound.
That is, the bounding problem is a form of quantitative information
flow {\em checking} problem (as opposed to {\em inference}).  Much of
the previous research has focused on information theoretic properties
of quantitative information flow and approximate (i.e., incomplete
and/or unsound) algorithms for checking and inferring quantitative
information flow.  To fill the void, in a recent
work~\cite{DBLP:conf/csfw/yasuoka2010}, we have studied the hardness
and possibilities of deciding the {\em comparison problem} of
quantitative information flow, which is the problem of precisely
checking if the information flow of one program is larger than that of
the other, that is, the problem of deciding if $\mathcal{X}(M_1) \leq
\mathcal{X}(M_2)$ given programs $M_1$ and $M_2$.  The study has lead
to some remarkable results, summarized in Section~\ref{sec:ksafety}
and Section~\ref{sec:complex} of this paper to contrast with the new
results on the bounding problem.  However, the hardness results on the
comparison problem do not imply hardness of the bounding
problem.\footnote{But, they imply the hardness of the inference
  problem because we can compare $\mathcal{X}(M_1)$ and
  $\mathcal{X}(M_2)$ once we have computed them.  We also note
    that the hardness of the bounding problems implies that of the
    comparison problems because we can reduce the bounding problem
    $\mathcal{X}(M) \leq q$ to a comparison problem that compares $M$
    with a program whose information flow is $q$.  (But, the reverse
    direction does not hold.)}  Thus, this paper settles the open
question.

We summarize the main results of the paper below.  Here, $\mathcal{X}$
is ${\it SE}[U]$, ${\it ME}[U]$, ${\it GE}[U]$ or ${\it CC}$, where $U$
is the uniform distribution.
\begin{itemize}
\item Checking if $\mathcal{X}(M) \leq q$ is not a $k$-safety
  property~\cite{terauchi:sas05,DBLP:conf/csfw/ClarksonS08} for any $k$.
\item Restricted to loop-free boolean programs, checking if
$\mathcal{X}(M) \leq q$ is PP-hard.
\end{itemize}
Roughly, a verification problem being $k$-safety means that it can be
reduced to a standard safety problem, such as the unreachability
problem, via self composition~\cite{barthe:csfw04,darvas:spc05}.  For
instance, non-interference is a $2$-safety property (technically, for
the termination-insensitive case\footnote{We restrict to terminating
  programs in this paper.  (The termination assumption is
  nonrestrictive because we assume safety verification as a blackbox
  routine.)}), and this has enabled its precise checking via a
reduction to a safety problem via self composition and applying
automated safety verification
techniques~\cite{terauchi:sas05,naumann:esorics06,unno:plas2006}.
Also, our recent work~\cite{DBLP:conf/csfw/yasuoka2010} has shown that
deciding the comparison problem of quantitative information flow for
all distributions (i.e., checking if $\forall \mu.  {\it SE}[\mu](M_1)
\leq {\it SE}[\mu](M_2)$, $\forall \mu.  {\it ME}[\mu](M_1) \leq {\it
  ME}[\mu](M_2)$, $\forall \mu.  {\it GE}[\mu](M_1) \leq {\it
  GE}[\mu](M_2)$, and $\forall \mu.\forall h,\ell. {\it
  BE}[\aseq{\mu,h,\ell}](M_1) \leq {\it
  BE}[\aseq{\mu,h,\ell}](M_2)$\footnote{See below for the notation
  ${\it BE}[\aseq{\mu,h,\ell}](M)$ denoting the belief-based
  quantitative information flow of $M$ with respect to the experiment
  $\aseq{\mu,h,\ell}$.  The result for the belief-based definition is
  proven in the extended version of the paper that is under
  submission~\cite{yasuoka:toplas2010submit}.}) are $2$-safety problems
(and in fact, all equivalent).

We also prove a complexity theoretic gap with these related problems.
We have shown in the previous paper~\cite{DBLP:conf/csfw/yasuoka2010}
that, for loop-free boolean programs, both checking non-interference
and the above comparison problem with universally quantified
distributions are coNP-complete.  (PP is believed to be strictly
harder than coNP.  In particular, $\text{coNP} = \text{PP}$ implies
the collapse of the polynomial hierarchy to level 1.)

Therefore, the results suggest that the bounding problems of
quantitative information flow are harder than the related problems of
checking non-interference and the quantitative information flow
comparison problems with universally quantified distributions, and may
require different techniques to solve (i.e., not self composition).

The belief-based quantitative information flow~\cite{clarkson:csf2005}
differs from the definitions above in that it focuses on the
information flow from a particular execution of the program (called
{\em experiment}) rather than the information flow from all executions
of the program.\footnote{Clarkson et. al.~\cite{clarkson:csf2005}
    also propose a definition which averages the quantitative
    information flow over a distribution of the inputs $h$ and $\ell$.
    Note that a hardness result for (1) below implies the hardness
    result of the bounding problem for this problem as we may take the
    distribution to be a point mass.}  Therefore, we define and study
the hardness of two types of bounding problems for the belief-based
definition:
\begin{itemize}
\item[(1)] ${\it BE}[\aseq{\mu,h,\ell}](M)\le q$
\item[(2)] $\forall h,\ell.{\it BE}[\aseq{\mu,h,\ell}](M)\le q$
\end{itemize}
Here, ${\it BE}[\aseq{\mu,h,\ell}](M)$ denotes the belief-based
information flow of $M$ with the experiment $\aseq{\mu,h,\ell}$ where
$h,\ell$ are the particular (high-security and low-security) inputs.
Note that the problem (2) checks the bound of the belief-based
quantitative information flow for {\em all} inputs whereas (1) checks
the information flow for a particular input.  This paper proves that
neither of these problems are $k$-safety for any $k$, and are PP-hard
for loop-free boolean programs.

We note that the above results are for the case the quantity $q$ is
taken to be an input to the bounding problems.  We show that when
fixing the parameter $q$ constant, some of the problems become
$k$-safety under certain conditions for different $k$'s
(cf. Section~\ref{sec:ksafetyconst}, \ref{sec:ksafetybelief}, and
\ref{sec:ksafetycclike}).

We also define and study the hardness of the following bounding
problems that check the bound over {\em all} distributions.
\begin{itemize}
\item[(1)] $\forall \mu.{\it SE}[\mu](M)\le q$
\item[(2)] $\forall \mu.{\it ME}[\mu](M)\le q$
\item[(3)] $\forall \mu.{\it GE}[\mu](M)\le q$
\item[(4)] $\forall \mu.{\it BE}[\aseq{\mu,h,\ell}](M)\le q$
\item[(5)] $\forall \mu,h,\ell.{\it BE}[\aseq{\mu,h,\ell}](M)\le q$
\end{itemize}
We show that except for (4) and (5), these problems are also not
$k$-safety for any $k$, and are PP-hard for loop-free boolean
programs, when $q$ is not a constant (but are $k$-safety for various
$k$'s when $q$ is held constant).  For the problems (4) and (5), we
show that the problems are actually equivalent to that of checking
non-interference.  (1), (2), and (3) are proven by showing that
the problems correspond to various ``channel capacity like'' definitions
of quantitative information flow.

The rest of the paper is organized as follows.
Section~\ref{sec:prelim} reviews the existing information-theoretic
definitions of quantitative information flow and formally defines the
bounding problems.  Section~\ref{sec:ksafety} proves that the bounding
problems are not $k$-safety problems for ${\it SE}[U]$, ${\it ME}[U]$,
${\it GE}[U]$, and ${\it CC}$.  (Section~\ref{sec:ksafetyconst} shows
that when fixing the parameter $q$ constant, some of them become
$k$-safety under certain conditions for different $k$'s.)
Section~\ref{sec:ksafetybelief} shows $k$-safety results for the
belief-based bounding problems, and Section~\ref{sec:ksafetycclike}
shows $k$-safety results for the bounding problems that check the
bound for all distributions.  Section~\ref{sec:complex} proves
complexity theoretic hardness results for the bounding problems for
loop-free boolean programs for ${\it SE}[U]$, ${\it ME}[U]$, ${\it
  GE}[U]$, and ${\it CC}$, and Section~\ref{sec:complexbeliefcclike}
proves those for the belief-based bounding problems and the bounding
problems that check the bound for all distributions.
Section~\ref{sec:discussion} discusses some implications of the
hardness results.  Section~\ref{sec:related} discusses related work,
and Section~\ref{sec:concl} concludes.  All the proofs appear in
Appendix~\ref{appendix}.

\section{Preliminaries}

\label{sec:prelim}

We introduce the information theoretic definitions of quantitative
information flow that have been proposed in literature.  First, we
review the notion of the {\em Shannon entropy}~\cite{shannon48},
$\mathcal{H}[\mu](X)$, which is the average of the information
content, and intuitively, denotes the uncertainty of the random
variable $X$.
\begin{definition}[Shannon Entropy]
  Let $X$ be a random variable with sample space $\mathbb X$ and $\mu$
  be a probability distribution associated with $X$.  (We write $\mu$
  explicitly for clarity.)  The Shannon entropy of $X$ is defined as
\[
\mathcal{H}[\mu](X)=\sum_{x\in\mathbb{X}} \mu(X=x)\log\frac{1}{\mu(X=x)}
\]
(The logarithm is in base 2.)
\end{definition}
Next, we define {\em conditional entropy}.  Informally, the conditional
entropy of $X$ given $Y$ denotes the uncertainty of $X$ after knowing
$Y$.
\begin{definition}[Conditional Entropy]
Let $X$ and $Y$ be random variables with sample spaces $\mathbb X$ and
$\mathbb Y$, respectively, and $\mu$ be a probability distribution
associated with $X$ and $Y$.  Then, the conditional entropy of $X$
given $Y$, written $\mathcal{H}[\mu](X|Y)$ is defined as
\[
\mathcal{H}[\mu](X|Y)=\sum_{y\in\mathbb Y} \mu(Y=y) \mathcal{H}[\mu](X|Y=y)
\]
where
\[
\begin{array}{l}
\mathcal{H}[\mu](X|Y=y)
=\sum_{x\in\mathbb X} \mu(X=x|Y=y)\log\frac{1}{\mu(X=x|Y=y)} \\
\mu(X=x|Y=y)=\frac{\mu(X=x,Y=y)}{\mu(Y=y)}
\end{array}
\]
\end{definition}
Next, we define (conditional) mutual information.  Intuitively, the
conditional mutual information of $X$ and $Y$ given $Z$ represents the
mutual dependence of $X$ and $Y$ after knowing $Z$.
\begin{definition}[Mutual Information]
  Let $X, Y$ and $Z$ be random variables and $\mu$ be an associated
  probability distribution.\footnote{We abbreviate the sample spaces
    of random variables when they are clear from the context.}  Then,
  the conditional mutual information of $X$ and $Y$ given $Z$ is
  defined as
\[
\begin{array}{rcl}
\mathcal{I}[\mu](X;Y|Z)&=&\mathcal{H}[\mu](X|Z)-\mathcal{H}[\mu](X|Y,Z)\\
&=&\mathcal{H}[\mu](Y|Z)-\mathcal{H}[\mu](Y|X,Z)
\end{array}
\]
\end{definition}

Let $M$ be a program that takes a high security input $H$ and a low
security input $L$, and gives the low security output $O$.  For
simplicity, we restrict to programs with just one variable of each
kind, but it is trivial to extend the formalism to multiple variables
(e.g., by letting the variables range over tuples).  Also, for the
purpose of the paper, unobservable (i.e., high security) outputs are
irrelevant, and so we assume that the only program output is the low
security output.  Let $\mu$ be a probability distribution over the
values of $H$ and $L$.  Then, the semantics of $M$ can be defined by
the following probability equation. (We restrict to terminating
  deterministic programs in this paper.)
\[
\mu(O = o) = \sum_{\scriptsize \begin{array}{l}h,\ell \in \mathbb{H}, \mathbb{L}\\ M(h,\ell) = o\end{array}} \mu(H = h, L = \ell)
\]
Note that we write $M(h,\ell)$ to denote the low security output of
the program $M$ given inputs $h$ and $\ell$.  Now, we are ready to
introduce the Shannon-entropy based definition of quantitative
information flow (QIF)~\cite{denning82,clarkjcs2007,malacaria:popl2007}.
\begin{definition}[Shannon-Entropy-based QIF]
\label{def:se}
Let $M$ be a program with a high security input $H$, a low security input
$L$, and a low security output $O$.  Let $\mu$ be a distribution over
$H$ and $L$.  Then, the Shannon-entropy-based quantitative information
flow is defined
\[
\begin{array}{rcl}
{\it SE}[\mu](M) & = & \mathcal{I}[\mu](O;H|L) \\
 & = & \mathcal{H}[\mu](H|L)-\mathcal{H}[\mu](H|O,L)
\end{array}
\]
\end{definition}
Intuitively, $\mathcal{H}[\mu](H|L)$ denotes the initial uncertainty
knowing the low security input and $\mathcal{H}[\mu](H|O,L)$ denotes
the remaining uncertainty after knowing the low security output.

\begin{sloppypar}
As an example, consider the programs $M_1$ and $M_2$ from
Section~\ref{sec:introduction}.  For concreteness, assume that $g$ is
the value $01$ and $H$ ranges over the space $\aset{00,01,10,11}$.
Let $U$ be the uniform distribution over $\aset{00,01,10,11}$, that
is, $U(h) = 1/4$ for all $h \in \aset{00,01,10,11}$.  Computing
their Shannon-entropy based quantitative information flow, we have,
\end{sloppypar}
\[
\begin{array}{l}
{\it SE}[U](M_1)=\mathcal{H}[U](H)-\mathcal{H}[U](H|O)
=\log 4-\frac{3}{4}\log{3}
\approx .81128\\

{\it SE}[U](M_2)=\mathcal{H}[U](H)-\mathcal{H}[U](H|O)
=\log 4-\log 1
=2
\end{array}
\]
Hence, if the user was to ask if ${\it SE}[U](M_1) \leq 1.0$, that is,
``does $M_1$ leak more than one bit of information (according to ${\it
SE}[U]$)?'', then the answer would be no.  But, for the same query,
the answer would be yes for $M_2$.

Next, we introduce the {\em min entropy}, which Smith~\cite{smith09}
recently suggested as an alternative measure for quantitative information
flow.
\begin{definition}[Min Entropy]
Let $X$ and $Y$ be random variables, and $\mu$ be an associated probability
distribution.  Then, the min entropy of $X$ is defined
\[
\mathcal{H}_\infty[\mu](X)=\log\frac{1}{\mathcal{V}[\mu](X)}
\]
and the conditional min entropy of $X$ given $Y$ is defined
\[
\mathcal{H}_\infty[\mu](X|Y)=\log\frac{1}{\mathcal{V}[\mu](X|Y)}
\]
where
\[
\begin{array}{rcl}
\mathcal{V}[\mu](X)&=&\max_{x\in\mathbb X} \mu(X=x)\\
\mathcal{V}[\mu](X|Y=y)&=&\max_{x\in\mathbb X} \mu(X=x|Y=y)\\
\mathcal{V}[\mu](X|Y)&=&\sum_{y\in\mathbb Y} \mu(Y=y) \mathcal{V}[\mu](X|Y=y)
\end{array}
\]
\end{definition}

\begin{sloppypar}
Intuitively, $\mathcal{V}[\mu](X)$ represents the highest probability
that an attacker guesses $X$ in a single try.  We now define the
min-entropy-based definition of quantitative information flow.
\end{sloppypar}

\begin{definition}[Min-Entropy-based QIF]
\label{def:me}
Let $M$ be a program with a high security input $H$, a low security input
$L$, and a low security output $O$.  Let $\mu$ be a distribution over
$H$ and $L$.  Then, the min-entropy-based quantitative information
flow is defined
\[
{\it ME}[\mu](M)=\mathcal{H}_\infty[\mu](H|L)-\mathcal{H}_\infty[\mu](H|O,L)
\]
\end{definition}

Whereas Smith~\cite{smith09} focused on programs lacking low security
inputs, we extend the definition to programs with low security inputs
in the definition above.  It is easy to see that our definition
coincides with Smith's for programs without low security
inputs.  Also, the extension is arguably natural in the sense that we simply
take the conditional entropy with respect to the distribution over the
low security inputs.

Computing the min-entropy based quantitative information flow for our
running example programs $M_1$ and $M_2$ from
Section~\ref{sec:introduction} with the uniform distribution, we
obtain,
\[
\begin{array}{l}
{\it ME}[U](M_1)=\mathcal{H}_\infty[U](H)-\mathcal{H}_\infty[U](H|O)
=\log 4-\log 2
=1\\
{\it ME}[U](M_2)=\mathcal{H}_\infty[U](H)-\mathcal{H}_\infty[U](H|O)
=\log 4 -\log 1
=2
\end{array}
\]
Hence, if a user is to check whether ${\it ME}[U]$ is bounded by $q$
for $1 \leq q < 2$, then the answer would be yes for $M_1$, but no for
$M_2$.

Next, we introduce the {\em guessing-entropy} based definition of
quantitative information
flow~\cite{Massey94,kopf07,DBLP:conf/sp/BackesKR09}.
\begin{definition}[Guessing Entropy]
Let $X$ and $Y$ be random variables, and $\mu$ be an associated probability
distribution.  Then, the guessing entropy of $X$ is defined
\[
\mathcal{G}[\mu](X)=\sum_{1\le i\le m}i\times\mu(X=x_i)
\]
where $m=|\mathbb{X}|$ and $x_1,x_2,\dots,x_{m}$ satisfies $\forall
i,j.i\le j\Rightarrow \mu(X=x_i)\ge\mu(X=x_j)$.

The conditional guessing entropy of $X$ given $Y$ is defined
\[
\mathcal{G}[\mu](X|Y)=\sum_{y\in{\mathbb Y}}\mu(Y=y)\mathcal{G}[\mu](X|Y=y)\\
\]
where
\[
\begin{array}{c}
\mathcal{G}[\mu](X|Y=y)=\sum_{1\le i\le m}i\times\mu(X=x_i|Y=y)\\
m=|\mathbb{X}|\;\textrm{ and }\;\forall
i,j.i\le j\Rightarrow \mu(X=x_i|Y=y)\ge\mu(X=x_j|Y=y)
\end{array}
\]
\end{definition}

Intuitively, $\mathcal{G}[\mu](X)$ represents the average number of
times required for the attacker to guess the value of $X$.  We now
define the guessing-entropy-based quantitative information flow.

\begin{definition}[Guessing-Entropy-based QIF]
\label{def:ge}
Let $M$ be a program with a high security input $H$, a low security input
$L$, and a low security output $O$.  Let $\mu$ be a distribution over
$H$ and $L$.  Then, the guessing-entropy-based quantitative
information flow is defined
\[
{\it GE}[\mu](M)=\mathcal{G}[\mu](H|L)-\mathcal{G}[\mu](H|O,L)
\]
\end{definition}

\begin{sloppypar}
Like with the min-entropy-based definition, the previous research on
guessing-entropy-based quantitative information flow only considered
programs without low security
inputs~\cite{kopf07,DBLP:conf/sp/BackesKR09}.  But, it is easy to see
that our definition with low security inputs coincides with the
previous definitions for programs without low security inputs.  Also,
as with the extension for the min-entropy-based definition, it simply
takes the conditional entropy over the low security inputs.
\end{sloppypar}

We test {\it GE} on the running example from
Section~\ref{sec:introduction} by calculating the quantities for the
programs $M_1$ and $M_2$ with the uniform distribution. 
\renewcommand{\arraystretch}{1.1}
\[
\begin{array}{l}
{\it GE}[U](M_1)  =\mathcal{G}[U](H)-\mathcal{G}[U](H|O)
=\frac{5}{2} - \frac{7}{4}
 = 0.75\\
{\it GE}[U](M_2) = \mathcal{G}[U](H)-\mathcal{G}[U](H|O)
=\frac{5}{2} - 1
 = 1.5
\end{array}
\]
\renewcommand{\arraystretch}{1.0}

\noindent
Hence, if a user is to check whether ${\it GE}[U]$ is bounded by
$q$ for $0.75 \leq q < 1.5$, then the answer would be yes for $M_1$, but no
for $M_2$.

Next, we introduce the belief-based definition of quantitative
information flow \cite{clarkson:csf2005}.  The belief-based definition
computes the information leak from a single execution of the program,
called an {\em experiment}.

\begin{definition}[Experiment]
  Let $\mu$ be a distribution over a high-security input such that
  $\forall h.\mu(h)>0$, $h_\mathcal{E}$ be a high-security input, and
  $\ell_\mathcal{E}$ be a low-security input.  Then, the experiment
  $\mathcal{E}$ is defined to be the tuple $\langle
  \mu,h_\mathcal{E},\ell_\mathcal{E}\rangle$.\footnote{Clarkson et.
    al.~\cite{clarkson:csf2005} also include the output and the
    program itself as part of the experiment.  In this paper, an
    experiment consists solely of the input and the distribution.}
\end{definition}
Intuitively, the distribution $\mu$ represents the attacker's {\em belief}
about the user's high security input selection, $\ell_\mathcal{E}$
denotes the attacker's low-security input selection, and
$h_\mathcal{E}$ denotes the user's actual selection.  Then, the
belief-based quantitative information flow, which is the information flow
of individual experiments, is defined as follows.
\begin{definition}[Belief-based QIF]
\label{def:beliefqif}
Let $M$ be a program with a high security input, a low security input,
and a low security output.  Let $\mathcal{E}$ be an experiment such
that $\mathcal{E}=\langle \mu,h_\mathcal{E},\ell_\mathcal{E}\rangle$.
Then, the belief-based quantitative information flow is defined

\[ {\it BE}[\mathcal{E}](M)= D(\mu \rightarrow \dot{h_\mathcal{E}})-D(\mu |
o_\mathcal{E} \rightarrow \dot{h_\mathcal{E}}) 
\]where\[
\begin{array}{l}
  o_\mathcal{E}=M(h_\mathcal{E},\ell_\mathcal{E})\\

\dot{h}=\lambda h'. {\sf if}\;h=h'\;{\sf then}\;1\;{\sf else}\;0\\

\mu_{\ell_\mathcal{E}}(o_\mathcal{E})=\sum_{h\in \aset{h'\mid M(h',\ell_\mathcal{E}) =  o_\mathcal{E}}} \mu(h)\\

\mu | o_\mathcal{E}=\lambda h. {\sf
  if}\;M(h,\ell_\mathcal{E})=o_\mathcal{E}\;{\sf
  then}\;\frac{\mu(h)}{\mu_{\ell_\mathcal{E}}(o_\mathcal{E})}\;{\sf else}\;0\\

D(\mu\rightarrow \mu')=\sum_h \mu'(h)\log\frac{\mu'(h)}{\mu(h)}
\end{array}
\]
\end{definition}
Here, $D(\mu \rightarrow \mu')$ is the {\em relative entropy}
  (or, {\em distance}) of $\mu$ and $\mu'$, and quantifies the
difference between the two distributions.\footnote{Here, we
    follow \cite{clarkson:csf2005} and use the notation
    $D(\mu\rightarrow \mu')$ over the more standard notation $D(\mu'
    || \mu)$.}  Note that $\dot{h}$ denotes the point mass
distribution at $h$.  Intuitively, the belief-based quantitative
information flow expresses the difference between the attacker's
belief about the high security input and the output of the experiment.
It can be shown that ${\it BE}[\mathcal{E}](M)$ is equivalent to {\em
  self-information} (for $M$ deterministic), that is, the
  negative logarithm of the probability the event occurs (i.e., in
  this case, the output occurs).
\begin{lemma}
  Let $\mu$ be a belief, $h_\mathcal{E}$ be a high-security input,
  $\ell_\mathcal{E}$ be a low-security input.  Then, ${\it
    BE}[\langle
    \mu,h_\mathcal{E},\ell_\mathcal{E}\rangle](M)=-\log\Sigma_{h\in
    \aset{h'\mid
      M(h',\ell_\mathcal{E})=M(h_\mathcal{E},\ell_\mathcal{E})}}
  \mu(h)$.
\label{lem:be}
\end{lemma}

Computing the belief-based quantitative information flow for our
running example programs $M_1$ and $M_2$ from
Section~\ref{sec:introduction} with the uniform distribution, we
obtain,
\begin{itemize}
\item $h\in\aset{00,10,11}$
\[
  {\it BE}[\langle U,h\rangle](M_1)=-\log U(M_1(h))=-\log\frac{3}{4}\approx.41503
\]

\item $h=01$
\[
  {\it BE}[\langle U,h\rangle](M_1)=-\log U(M_1(h))=-\log\frac{1}{4}=2
\]
\end{itemize}
And, for any $h\in\aset{00,01,10,11}$,
\[
  {\it BE}[\langle U,h\rangle](M_2)=-\log U(M_2(h))=-\log\frac{1}{4}=2
\]
Therefore, if the user was to ask if ${\it BE}[\aseq{U,h}]$ is bounded
by $1.0$ for $h = 00$, then the answer would be yes for $M_1$ but no
for $M_2$.  But, if the user was to ask if ${\it BE}[\aseq{U,h}]$ is
bounded by $1.0$ for all $h$, then the answer would be no for both
$M_1$ and $M_2$.

Finally, we introduce the definition of quantitative information flow
based on {\em channel
  capacity}~\cite{mccamant:pldi2008,malacaria08,NMS2009}, which is
defined to be the maximum of the Shannon-entropy based quantitative
information flow over the distribution.
\begin{definition}[Channel-Capacity-based QIF]
Let $M$ be a program with a high security input $H$, a low security input
$L$, and a low security output $O$.  Then, the channel-capacity-based
quantitative information flow is defined
\[
{\it CC}(M)=\max_\mu\mathcal{I}[\mu](O;H|L)
\]
\end{definition}

Unlike the other definitions above, the channel-capacity based
definition of quantitative information flow is not parameterized by
the distribution over the inputs.  As with the other definitions, let
us test the definition on the running example from
Section~\ref{sec:introduction} by calculating the quantities for the
programs $M_1$ and $M_2$:
\[
\begin{array}{l}
{\it CC}(M_1)=\max_\mu\mathcal{I}[\mu](O;H) =1 \\
{\it CC}(M_2)=\max_\mu\mathcal{I}[\mu](O;H) =2
\end{array}
\]
Note that ${\it CC}(M_1)$ (resp. ${\it CC}(M_2)$) is equal to ${\it
ME}[U](M_1)$ (resp. ${\it ME}[U](M_2)$).  This is not a coincidence.
In fact, it is known that ${\it CC}(M) = {\it ME}[U](M)$ for all
programs $M$ without low security inputs~\cite{smith09}.

\subsection{Non-interference}

\label{sec:nonint}

We recall the notion of
non-interference~\cite{DBLP:conf/sosp/Cohen77,goguen:sp1982}.
\begin{definition}[Non-intereference]
A program $M$ is said to be non-interferent iff for any $h,h'\in
\mathbb{H}$ and $\ell\in\mathbb{L}$, $M(h,\ell)=M(h',\ell)$.
\end{definition}

It can be shown that for the definitions of quantitative information
flow $\mathcal{X}$ introduced above, $\mathcal{X}(M) \leq 0$ iff $M$
is non-interferent.\footnote{Technically, we need the non-zero-ness
  condition on the distribution.  (See below.)}  That is, the bounding
problem (which we only officially define for positive bounds --see
Section~\ref{sec:boundprob}--) degenerates to checking
non-interference when $0$ is given as the bound.

\begin{theorem}
\label{thm:nonint}
Let $\mu$ be a distribution such that $\forall
h\in\mathbb{H},\ell\in\mathbb{L}.\mu(h,\ell)>0$.  Then,
\begin{itemize}
\item $M$ is non-interferent if and only if ${\it SE}[\mu](M) \leq 0$.
\item $M$ is non-interferent if and only if ${\it ME}[\mu](M) \leq 0$.
\item $M$ is non-interferent if and only if ${\it GE}[\mu](M) \leq 0$.
\item $M$ is non-interferent if and only if ${\it BE}[\aseq{\mu',h,\ell}](M)\le 0$.\footnote{Recall Definition~\ref{def:beliefqif} that $\mu'$ is a distribution over $\mathbb{H}$ such that $\mu'(h) > 0$ for all $h \in \mathbb{H}$.}
\item $M$ is non-interferent if and only if ${\it CC}(M) \leq 0$.
\end{itemize}
\end{theorem}
The equivalence result on the Shannon-entropy-based definition is
proven by Clark et al.~\cite{clark05}.  The proofs for the other four
definitions are given in Appendix~\ref{appendix}.

\subsection{Bounding Problem}

\label{sec:boundprob}

We define the {\em bounding problem} of quantitative information flow
for each definition introduced above.  The bounding problem for the
Shannon-entropy based definition $B_{\it SE}[\mu]$ is defined as
follows: Given a program $M$ and a positive real number $q$, decide if
${\it SE}[\mu](M) \leq q$.\footnote{Note that we treat $\mu$ as a
    parameter of the bounding problem rather than as an input.}
Similarly, we define the bounding problems for the other three
definitions $B_{\it ME}[\mu]$, $B_{\it GE}[\mu]$, and $B_{\it CC}$ as
follows.
\[
\begin{array}{rcl}
B_{\it ME}[\mu]&=&\aset{(M,q)\mid {\it ME}[\mu](M)\le q}\\
B_{\it GE}[\mu]&=&\aset{(M,q)\mid {\it GE}[\mu](M)\le q}\\ 
B_{\it CC}&=&\aset{(M,q)\mid {\it CC}(M)\le q}
\end{array}
\]
We defer the definitions of the belief-based bounding problems to 
Section~\ref{sec:ksafetybelief}.
\section{K-Safety Property}

\label{sec:ksafety}

We show that none of the bounding problems are $k$-safety problems for
any $k$.  Informally, a program property is said to be a {\em
  $k$-safety}
property~\cite{terauchi:sas05,DBLP:conf/csfw/ClarksonS08} if it can be
refuted by observing $k$ number of (finite) execution traces.  A
$k$-safety problem is the problem of checking a $k$-safety property.
Note that the standard safety property is a $1$-safety property.  An
important property of a $k$-safety problem is that it can be reduced
to a standard safety (i.e., $1$-safety) problem, such as the
unreachability problem, via a simple program transformation called
{\em self composition}~\cite{barthe:csfw04,darvas:spc05}.  This allows
one to verify $k$-safety problems by applying powerful automated
safety verification
techniques~\cite{DBLP:conf/popl/BallR02,DBLP:conf/popl/HenzingerJMS02,mcmillan:cav06,DBLP:journals/sttt/BeyerHJM07}
that have made remarkable progress recently.

As stated earlier, we prove that no bounding problem is a $k$-safety
property for any $k$.  (First, we prove the result for {\it SE}, {\it
  ME}, {\it GE}, and {\it CC}, and defer the result for {\it BE} to
Section~\ref{sec:ksafetybelief}.)  To put the result in perspective, we
compare it to the results of the related problems, summarized below.
Here, $\mathcal{X}$ is ${\it SE}[U]$, ${\it ME}[U]$, ${\it GE}[U]$, or
${\it CC}$, and $\mathcal{Y}$ is ${\it SE}$, ${\it ME}$, or ${\it
  GE}$.  (Recall that $U$ denotes the uniform distribution.)
\begin{itemize}
\item[(1)] Checking non-interference is a $2$-safety problem, but it is not $1$-safety.
\item[(2)] Checking $\mathcal{X}(M_1) \leq \mathcal{X}(M_2)$ is not a $k$-safety problem for any $k$.
\item[(3)] Checking $\forall \mu.\mathcal{Y}[\mu](M_1) \leq \mathcal{Y}[\mu](M_2)$ is a $2$-safety problem.
\end{itemize}
The result (1) on non-interference is classic (see, e.g.,
\cite{mclean:sp94,barthe:csfw04,darvas:spc05}).  The results (2) and
(3) on comparison problems are proven in our recent
paper~\cite{DBLP:conf/csfw/yasuoka2010}.  Therefore, this section's
results imply that the bounding problems are harder to verify (at
least, via the self-composition approach) than non-interference and
the quantitative information flow comparison problems with universally
quantified distributions.

Let ${\it Prog}$ be the set of all programs, and ${\mathbb R}^+$
  be the set of positive real numbers.  Let $\sembrack{M}$ denote the
  semantics (i.e., traces) of $M$, represented by the set of
  input/output pairs, that is, $\sembrack{M} = \aset{((h,\ell),o) \mid
    h \in \mathbb{H}, \ell \in \mathbb{L}, o = M(h,\ell)}$.  Then,
  formally, $k$-safety property is defined as follows.
\begin{definition}[$k$-safety property]
We say that a property $P \subseteq {\it Prog}\times {\mathbb R}^+$
is a $k$-safety property iff $(M,q)\not\in P$ implies that there exists
$T \subseteq \sembrack{M}$ such that $|T|\le k$ and $\forall
M'. T\subseteq\sembrack{M'}\Rightarrow(M',q)\not\in P$.
\end{definition}
Note that the original definition of $k$-safety property is only
defined over
programs~\cite{terauchi:sas05,DBLP:conf/csfw/ClarksonS08}.  However,
because the bounding problems take the additional input $q$, we extend
the notion to account for the extra parameter.

We now state the main results of this section which show that none of
the bounding problems are $k$-safety problems for any $k$.  Because we
are interested in hardness, we focus on the case where the
distribution is the uniform distribution.  That is, the results we
prove for the specific case applies to the general case.
\begin{theorem}
Neither $B_{\it SE}[U]$, $B_{\it ME}[U]$, $B_{\it GE}[U]$, nor $B_{\it
  CC}$ is a k-safety property for any k such that $k > 0$.
\label{thm:senk2}
\end{theorem}
The result follows from the fact that for each of bounding
  problem $B_\mathcal{X}$ above, for any $k$, there exists $q$ such
  that deciding $(M, q) \in B_\mathcal{X}$ is not a $k$-safety
  property.  In fact, as we show next, for some of the problems such
  as $B_{\it SE}[U]$, even if we fix $q$ to an arbitrary constant,
  there exists no $k$ such that the problem is $k$-safety.  (But for
  other problems, for certain cases, we can find $k$ that depends on
  $q$.)  We defer the details to the next section.  (See also
  Section~\ref{sec:lowsecinputs}.)
\subsection{K-Safety Under a Constant Bound}

\label{sec:ksafetyconst}

The result above appears to suggest that the bounding problems are
equally difficult for ${\it SE}[U]$, ${\it ME}[U]$, ${\it GE}[U]$, and
$CC$.  However, holding the parameter $q$ constant (rather than having
it as an input) paints a different picture.  We show that the problems
become $k$-safety for different definitions for different $k$'s under
different conditions in this case.

First, for $q$ fixed, we show that the bounding problem for the
channel-capacity based definition of quantitative information flow is
$k$-safety for $k = \lfloor2^q\rfloor + 1$.  (Also, this bound is
tight.)
\begin{theorem}
\label{thm:cck}
Let $q$ be a constant.  Then, $B_{\it CC}$ is
$\lfloor2^q\rfloor+1$-safety, but it is not $k$-safety for any $k \leq
\lfloor2^q\rfloor$.
\end{theorem}

We briefly explain the intuition behind the above result. Recall that
a problem being $k$-safety means the existence of a {\em
  counterexample} trace set of size at most $k$.  That is, for $(M, q)
\notin B_{\it CC}$, we have $T \subseteq \sembrack{M}$ such that $|T|
\leq \lfloor2^q\rfloor+1$ such that any program that also contains $T$
as its traces also does not belong to $B_{\it CC}$ (with $q$), that
is, its channel-capacity-based quantitative information flow is
greater than $q$.  Then, the above result follows from the fact that
the channel-capacity-based quantitative information flow coincides
with the maximum over the low security inputs of the logarithm of the
number of outputs~\cite{malacaria08}, therefore, any $T$ containing
$\lfloor 2^q\rfloor +1$ traces of the same low security input and
disjoint outputs is a counterexample.

For concreteness, we show how to check $B_{\it CC}$ via self
composition.  Suppose we are given a program $M$ and a positive real
$q$.  We construct the self-composed program $M'$ shown below.
\[
\begin{array}{l}
  M'(H_1,H_2,\dots,H_n,L)\equiv\\
  \quad O_1:=M(H_1,L);O_2:=M(H_2,L);\dots;O_n:=M(H_n,L);\\
  \quad {\sf assert}(\bigvee_{i,j\in\aset{1,\dots,n}}(O_i=O_j\wedge i\not=j))
\end{array}
\]
where $n=\lfloor 2^q\rfloor +1$.  In general, a self composition
involves making $k$ copies the original program so that the resulting
program would generate $k$ traces of the original (having the desired
property).  By the result proven by Malacaria and
Chen~\cite{malacaria08}(see also Lemma~\ref{lem:ccloglow}), it follows
that $M'$ does not cause an assertion failure iff $(M,q)\in B_{\it
  CC}$.

Next, we show that for programs without low security inputs, $B_{\it
ME}[U]$ and $B_{\it GE}[U]$ are also both $k$-safety problems (but for
different $k$'s) when $q$ is held constant.
\begin{theorem}
\label{thm:mek}
Let $q$ be a constant, and suppose $B_{\it ME}[U]$ only takes programs
  without low security inputs.  Then, $B_{\it ME}[U]$ is $\lfloor
  2^q\rfloor +1$-safety, but it is not $k$-safety for any $k \leq
  \lfloor2^q\rfloor$.
\end{theorem}
\begin{theorem}
\label{thm:gek}
Let $q$ be a constant, and suppose $B_{\it GE}[U]$ only takes programs
without low security inputs.  If $q\ge\frac{1}{2}$, then, $B_{\it
  GE}[U]$ is $\lfloor\frac{(\lfloor q\rfloor +1)^2}{\lfloor q\rfloor
  +1 -q} \rfloor +1$-safety, but it is not $k$-safety for any $k \leq
\lfloor\frac{(\lfloor q\rfloor +1)^2}{\lfloor q\rfloor +1 -q}
\rfloor$.  Otherwise, $q<\frac{1}{2}$ and $B_{\it GE}[U]$ is
$2$-safety, but it is not $1$-safety.
\end{theorem}

The result for ${\it ME}[U]$ follows from the fact that for programs
without low security inputs, the min-entropy based quantitative
information flow with the uniform distribution is actually equivalent
to the channel-capacity based quantitative information
flow~\cite{smith09}.  The result for ${\it GE}[U]$ may appear less
intuitive, but, the key observation is that, like the channel-capacity
based definition and the min-entropy based definition with the uniform
distribution (for the case without low security inputs), for any set
of traces $T = \sembrack{M}$, the information flow of a program
containing $T$ would be at least as large as that of $M$.  Therefore,
by holding $q$ constant, we can always find a large enough
counterexample $T$.  The reason $B_{\it GE}[U]$ is $2$-safety for
$q<\frac{1}{2}$ is because, in the absence of low security inputs, the
minimum non-zero quantity of ${\it GE}[U](M)$ is bounded (by $1/2$),
and so for such $q$, the problem ${\it GE}[U](M) \leq q$ is equivalent
to checking non-interference.\footnote{In fact, the minimum non-zero
  quantity property also exists for {\it ME}[U] without low security
  inputs and {\it CC}.  There, the minimum non-zero quantity is $1$,
  which agrees with the formulas given in the theorems.}

But, when low security inputs are allowed, neither $B_{\it ME}[U]$ nor
$B_{\it GE}[U]$ are $k$-safety for any $k$, even when $q$ is held
constant.
\begin{theorem}
\label{thm:menk}
Let $q$ be a constant.  (And let $B_{\it ME}[U]$ take programs with
low security inputs.) Then, $B_{\it ME}[U]$ is not a $k$-safety
property for any $k > 0$.
\end{theorem}
\begin{theorem}
\label{thm:genk}
Let $q$ be a constant.  (And let $B_{\it GE}[U]$ take programs with
low security inputs.) Then, $B_{\it GE}[U]$ is not a $k$-safety
property for any $k > 0$.
\end{theorem}

Finally, we show that the Shannon-entropy based definition (with the
uniform distribution) is the hardest of all the definitions and show
that its bounding problem is not a $k$-safety property for any $k$,
with or without low-security inputs, even when $q$ is held constant.
\begin{theorem}
\label{thm:senk}
Let $q$ be a constant, and suppose $B_{\it SE}[U]$ only takes programs
without low security inputs. Then, $B_{\it SE}[U]$ is not a $k$-safety
property for any $k > 0$.
\end{theorem}

Intuitively, Theorems~\ref{thm:menk}, \ref{thm:genk}, and
\ref{thm:senk} follow from the fact that, for these definitions, given
any potential counterexample $T \subseteq \sembrack{M}$ to show $(M,q)
\notin B_\mathcal{X}$, it is possible to find $M'$ containing $T$
whose information flow is arbitrarily close to $0$ (and so $(M',q) \in
B_\mathcal{X}$).  See Section~\ref{sec:lowsecinputs} for further
discussion.

Because $k$ tends to grow large as $q$ grows for all the definitions
and it is impossible to bound $k$ for all $q$, this section's results
are unlikely to lead to a practical verification of quantitative
information flow.  \footnote{But, a recent
  work~\cite{DBLP:conf/acsac/Heusser2010} shows some promising
  results.}  Nevertheless, the results reveal interesting disparities
among the different proposals for the definition of quantitative
information flow.

\subsection{K-Safety for Belief-based Definition}

\label{sec:ksafetybelief}

This section investigates the hardness of the bounding problems for
the belief-based definition of quantitative information flow.  We
define two types of bounding problems.
\[
\begin{array}{rcl}
B_{\it BE1}[\aseq{\mu,h,\ell}]&=&\aset{(M,q)\mid {\it
BE}[\aseq{\mu,h,\ell}](M)\le q}\\ 
B_{\it BE2}[\mu]&=&\aset{(M,q)\mid \forall h,\ell.{\it
BE}[\aseq{\mu,h,\ell}](M)\le q}
\end{array}
\]
$B_{\it BE1}$ checks the program's information flow against the given
quantity for a specific input pair $h,\ell$ whereas $B_{\it BE2}$
checks that for all inputs.

We show that these problems are not a $k$-safety problems for any $k$,
at least when $q$ is not a constant.  To put the result in
perspective, we compare to the results of the comparison problem for
the belief-based quantitative information flow
problem~\cite{yasuoka:toplas2010submit}.
\begin{itemize}
\item[(1)] Checking ${\it BE}[\aseq{U,h,\ell}](M_1) \leq{\it BE}[\aseq{U,h,\ell}](M_2)$ is not a $k$-safety problem for any $k$.
\item[(2)] Checking $\forall h,\ell.{\it BE}[\aseq{U,h,\ell}](M_1) \leq{\it BE}[\aseq{U,h,\ell}](M_2)$ is not a $k$-safety problem for any $k$.
\item[(3)] Checking $\forall \mu,h,\ell.{\it BE}[\aseq{\mu,h,\ell}](M_1) \leq{\it BE}[\aseq{\mu,h,\ell}](M_2)$ is a $2$-safety problem.
\end{itemize}
Note that the problem in (3) compares the two programs for {\em all}
experiments $\aseq{\mu,h,\ell}$.  This problem also turns out to
be equivalent to the comparison problems with universally quantified
distributions for {\it SE}, {\it ME}, and {\it GE} discussed in
Section~\ref{sec:ksafety}.  Hence, this section's non-$k$-safety
results show that the bounding problems $B_{\it BE1}$ and $B_{\it
  BE2}$ are harder to verify (at least, via the self-composition
approach) than non-interference and the comparison problems with
universally quantified distributions and experiments.

First, we show that $B_{\it BE1}[\aseq{U,h,\ell}]$ is not a $k$-safety
property for any $k$, even when $q$ is held constant, and even without
low security inputs.
\begin{theorem}
\label{thm:be1nk}
Let $q$ be a constant, and suppose $B_{\it BE1}[\aseq{U,h}]$ only
takes programs without low security inputs.  Then, $B_{\it
  BE1}[\aseq{U,h}]$ is not a $k$-safety property for any $k > 0$.
\end{theorem}
Next, we show that $B_{\it BE2}[U]$ is also not a $k$-safety property
for any $k$ when $q$ is a constant and $q \geq 1$, even without low
security inputs.  But, when $q$ is held constant and $q < 1$, $B_{\it
  BE2}[U]$ is a $2$-safety property.
\begin{theorem}
\label{thm:be2nk1}
  Let $q$ be a constant.  If $q\ge 1$, then $B_{\it BE2}[U]$ is not a
  $k$-safety property for any $k > 0$ even when $B_{\it BE2}[U]$ only
  takes programs without low security inputs.  Otherwise, $q<1$ and
  $B_{\it BE2}[U]$ is a 2-safety property, but it is not a 1-safety
  property.
\end{theorem}
The $2$-safety property for the case $q < 1$ follows because $B_{\it
  BE2}[U]$ turns out to be equivalent to non-interference for such
$q$.  The results show that the bounding problems for the belief-based
definition is also quite hard, except for the case where one checks if
the information flow is less than $1$ for all inputs, which degenerates
to checking non-interference.

\subsection{K-Safety for Channel Capacity Like Definitions}

\label{sec:ksafetycclike}

In this section, we study the hardness of the bounding problems that
check the bound for all distributions.  We define the following problems.
\[
\begin{array}{lcl}
B_{\it SECC}&=&\aset{(M,q)\mid\forall \mu. {\it SE}[\mu](M)\le q}\\
B_{\it MECC}&=&\aset{(M,q)\mid\forall \mu. {\it ME}[\mu](M)\le q}\\
B_{\it GECC}&=&\aset{(M,q)\mid\forall \mu. {\it GE}[\mu](M)\le q} \\
B_{\it BE1CC}[h,\ell]&=&\aset{(M,q)\mid\forall \mu. {\it BE}[\aseq{\mu,h,\ell}](M)\le q} \\
B_{\it BE2CC}&=&\aset{(M,q)\mid\forall \mu.\forall h,\ell. {\it BE}[\aseq{\mu,h,\ell}](M)\le q}
\end{array}
\]
Note that $B_{\it SECC} = B_{\it CC}$ because ${\it CC}(M) = \max_\mu
{\it SE}[\mu](M)$.  For this reason, we call these bounding
problems ``channel capacity like.'' For instance, K{\"o}pf and
Smith~\cite{DBLP:conf/csfw/KopfS10} call $\max_\mu {\it ME}[\mu](M)$
the {\em min-entropy channel capacity}.  (Note that $(M,q) \in B_{\it
  MECC}$ iff $\max_\mu {\it ME}[\mu](M) \leq q$.)  $B_{\it GECC}$
follows the same spirit.  We define two types of channel-capacity like
problems for the belief-based definition corresponding to the two
types of bounding problems $B_{\it BE1}$ and $B_{\it BE2}$.

We prove $k$-safety results for each of these problems.  The result
below for $B_{\it SECC}$ follows directly from that of $B_{\it CC}$
(i.e., Theorem~\ref{thm:cck}).  But, the other results proved are
  new.
\begin{theorem}
\label{thm:secck}
Let $q$ be a constant.  Then, $B_{\it SECC}$ is
$\lfloor2^q\rfloor+1$-safety, but it is not $k$-safety for any $k \leq
\lfloor2^q\rfloor$.
\end{theorem}

First, we show that $B_{\it MECC}$ enjoys the same property as $B_{\it
  SECC}$.  That is, when $q$ is held constant, it is
$\lfloor2^q\rfloor+1$-safety, but it is not $k$-safety for any $k \leq
\lfloor2^q\rfloor$.  Note that unlike $B_{\it ME}[U]$, this holds even
for programs with low security inputs.  We show this by proving the
following lemma stating that $\max_\mu {\it ME}[\mu]$ is actually
equivalent to ${\it CC}(M)$.
\begin{lemma}
\label{lem:mecceqcc}
$\max_\mu {\it ME}[\mu](M)={\it CC}(M)$
\end{lemma}
The lemma extends the result by Braun et al.~\cite{Braun:09:MFPS} that
shows the equivalence for the low-security-input-free case.  By the
lemma, the $k$-safety result for $B_{\it MECC}$ follows directly from
that of $B_{\it CC}$.
\begin{theorem}\label{thm:mecck}
  Let $q$ be a constant.  Then, $B_{\it MECC}$ is
  $\lfloor2^q\rfloor+1$-safety, but it is not $k$-safety for any $k
  \leq \lfloor2^q\rfloor$.
\end{theorem}

Next, we prove that, when $q$ is held constant, $B_{\it GECC}$ is
$k$-safety for $k = \lfloor\frac{(\lfloor q\rfloor +1)^2}{\lfloor
  q\rfloor +1 -q} \rfloor +1$ when $q\geq\frac{1}{2}$ and is
$2$-safety for $q < \frac{1}{2}$.  Recall that these $k$-safety bounds
are equivalent to those of $B_{\it GE}[U]$ without low security inputs
(cf.~Theorem~\ref{thm:gek}).  However, unlike $B_{\it GE}[U]$, the
$k$-safety result here holds even for programs with low security
inputs.
\begin{theorem}\label{thm:gecck}
  Let $q$ be a constant.  If $q\ge\frac{1}{2}$, then, $B_{\it GECC}$
  is $\lfloor\frac{(\lfloor q\rfloor +1)^2}{\lfloor q\rfloor +1 -q}
  \rfloor +1$-safety, but it is not $k$-safety for any $k \leq
  \lfloor\frac{(\lfloor q\rfloor +1)^2}{\lfloor q\rfloor +1 -q}
  \rfloor$.  Otherwise, $q<\frac{1}{2}$ and $B_{\it GECC}$ is
  $2$-safety, but it is not $1$-safety.
\end{theorem}
\begin{sloppypar}
The above is shown by proving the following lemma which states that the
``guessing entropy channel capacity'' $\max_\mu {\it GE}[\mu]$ is
actually equivalent to $\max_\ell {\it GE}[U\otimes\dot{\ell}]$.
(See below for the definition of $U\otimes\dot{\ell}$.)
\end{sloppypar}
\begin{lemma}
\label{lem:gecc}
We have $\max_\mu {\it GE}[\mu](M)=\max_{\ell'} {\it GE}[U\otimes
\dot{\ell'}](M)$ where $U\otimes \dot{\ell'}$ denotes $\lambda
h,\ell.{\sf if}\;\ell=\ell'\;{\sf then}\; U(h)\;{\sf else}\;0$.
\end{lemma}

Finally, we prove somewhat surprising results for $B_{\it
  BE1CC}[h,\ell]$ and $B_{\it BE2CC}$ stating that they are in fact
equivalent to non-interference, independent of $q$.  It follows that
these problems are $2$-safety but not $1$-safety.
\begin{theorem}
\label{thm:be3ni}
$(M,q) \in B_{\it BE1CC}[h,\ell]$ iff $M(\ell)$ is non-interferent.
\end{theorem}
Here, $M(\ell)=\lambda h. M(h,\ell)$.  That is, the theorem states
that, for any $q$, $(M, q) \in B_{\it BE1CC}[h,\ell]$ iff the
program $M$ restricted to the low security input $\ell$ is
non-interferent. (Note that checking non-interference at a fixed low
security input is also a $2$-safety property and is not a $1$-safety
property.)

An analogous result holds for $B_{\it BE2CC}$.
\begin{theorem}
\label{thm:be4ni}
$(M,q) \in B_{\it BE2CC}$ iff $M$ is non-interferent.
\end{theorem}
Clarkson et al.~\cite{DBLP:conf/csfw/ClarksonS08} also studies $B_{\it
  BE2CC}$, which they call $QL$ in their paper.\footnote{Technically,
  they allow an experiment to consist of a sequence of runs of the
  program whereas we restrict an experiment to a single run.}  They
state that the problem is a {\em hypersafety} property, which is a
superset of $k$-safety properties.\footnote{Informally, a property is
  a hypersafety if there exists a counterexample set of traces of {\em
    any} size.}

\section{Complexities for Loop-free Boolean Programs}

\label{sec:complex}

In this section, we analyze the computational complexity of the
bounding problems when the programs are restricted to loop-free
boolean programs.  We compare the complexity theoretic
hardness of the bounding problems with those of the related problems
for the same class of programs, as we have done with the $k$-safety
property of the problems.

That is, we compare against the comparison problems of quantitative
information flow and the problem of checking non-interference for
loop-free boolean programs.  The complexity results for these problems
are summarized below.  Here, $\mathcal{X}$ is ${\it SE}[U]$, ${\it
  ME}[U]$, ${\it GE}[U]$, or ${\it CC}$, and $\mathcal{Y}$ is ${\it
  SE}$, ${\it ME}$, or ${\it GE}$.
\begin{itemize}
\item[(1)] Checking non-interference is coNP-complete
\item[(2)] Checking $\mathcal{X}(M_1) \leq \mathcal{X}(M_2)$ is PP-hard.
\item[(3)] Checking $\forall \mu.\mathcal{Y}[\mu](M_1) \leq \mathcal{Y}[\mu](M_2)$ is coNP-complete.
\end{itemize}
The results (1) and (3) are proven in our recent
paper~\cite{DBLP:conf/csfw/yasuoka2010}.  The result (2) is proven in
the extended version of the paper~\cite{yasuoka:toplas2010submit} and
tightens our (oracle relative) \#P-hardness result from the conference
version~\cite{DBLP:conf/csfw/yasuoka2010}, which states that for each
$C$ such that $C$ is the comparison problem for ${\it SE}[U]$, ${\it
  ME}[U]$, ${\it GE}[U]$, or ${\it CC}$, we have $\text{\#P} \subseteq
\text{FP}^C$.  (Recall that the notation $\text{FP}^A$ means the
complexity class of function problems solvable in polynomial time with
an oracle for the problem $A$.)  \#P is the class of counting problems
associated with NP.  PP is the class of decision problems solvable in
probabilistic polynomial time.  PP is known to contain both coNP and
NP, $\text{PH} \subseteq \text{P}^\text{PP} =
\text{P}^\text{\#P}$~\cite{toda91}, and PP is believed to be strictly
larger than both coNP and NP.  (In particular, PP = coNP would imply
the collapse of the polynomial hierarchy (PH) to level 1.)

We show that, restricted to loop-free boolean programs, the bounding
problems for the Shannon-entropy-based, the min-entropy-based, and the
guessing-entropy-based definition of quantitative information flow with
the uniform distribution (i.e., ${\it SE}[U]$, ${\it ME}[U]$, and
${\it GE}[U]$) and the channel-capacity based definition (i.e., ${\it
  CC}$) are all PP-hard.  (The results for the belief-based definition
and the channel-capacity-like definitions appear in
Section~\ref{sec:complexbeliefcclike}.) The results strengthen the
hypothesis that the bounding problems for these definitions are quite
hard.  Indeed, they show that they are complexity theoretically harder
than non-interference and the comparison problems with the universally
quantified distributions for loop-free boolean programs, assuming that
coNP and PP are separate.

\begin{figure}[t]
\[
\begin{array}[t]{rcl}
  M&::=&x:=\psi\mid M_0 ; M_1 \\
  & \mid &{\sf if}\; \psi\;{\sf then}\; M_0 \;{\sf else}\; M_1\\
  \phi,\psi&::=&{\sf true}\mid x\mid \phi\wedge \psi\mid \neg \phi
\end{array}
\]
\caption{The syntax of loop-free boolean programs}
\label{fig:syntax}
\end{figure}

\begin{figure}[t]
\[
\begin{array}[t]{l}
\wpre{x:=\psi}{\phi}=\phi[\psi/x]\\
\wpre{{\sf if}\;\psi\;{\sf then}\;M_0\;{\sf else}\;M_1}{\phi}\\
\qquad=(\psi\Rightarrow \wpre{M_0}{\phi})\wedge(\neg \psi\Rightarrow \wpre{M_1}{\phi})\\
\wpre{M_0;M_1}{\phi}=\wpre{M_0}{\wpre{M_1}{\phi}}
\end{array}
\]
\caption{The weakest precondition for loop-free boolean programs}
\label{fig:wpsemantics}
\end{figure}

We define the syntax of loop-free boolean programs in
Figure~\ref{fig:syntax}.  We assume the usual derived formulas $\phi
\Rightarrow \psi$, $\phi = \psi$, $\phi \vee \psi$, and ${\sf false}$.
We give the usual weakest precondition semantics in
Figure~\ref{fig:wpsemantics}.

To adapt the information flow framework to boolean programs, we make
each information flow variable $H$, $L$, and $O$ range over functions
mapping boolean variables of its kind to boolean values.  For example,
if $x$ and $y$ are low security boolean variables and $z$ is a high
security boolean variable, then $L$ ranges over the functions
$\aset{x,y} \rightarrow \aset{{\sf false},{\sf true}}$, and $H$ and
$O$ range over $\aset{z} \rightarrow \aset{{\sf false},{\sf
    true}}$.\footnote{ We do not distinguish input boolean variables
  from output boolean variables.  But, a boolean variable can be made
  output-only by assigning a constant to the variable at the start of
  the program and made input-only by assigning a constant at the end.}
(Every boolean variable is either a low security boolean variable or a
high security boolean variable.)  We write $M(h,\ell) = o$ for an
input $(h,\ell)$ and an output $o$ if $(h,\ell) \models
\wpre{M}{\phi}$ for a boolean formula $\phi$ such that $o \models
\phi$ and $o' \not\models \phi$ for all output $o' \neq o$.  Here,
$\models$ is the usual logical satisfaction relation, using
$h,\ell,o$, etc.~to look up the values of the boolean variables.
(Note that this incurs two levels of lookup.)

As an example, consider the following program.
\[
\begin{array}{c}
M \;\equiv\; \ttassign{z}{x};\ttassign{w}{y};\ttif{x \wedge y}{\ttassign{z}{\neg z}}{\ttassign{w}{\neg w}}
\end{array}
\]
Let $x$, $y$ be high security variables and $z, w$ be low security
variables.  Then,
\[
\begin{array}{rcll}
{\it SE}[U](M) & = & 1.5\\
{\it ME}[U](M) & = & \log 3 \approx 1.5849625\\
\end{array}
\hspace{1em}
\begin{array}{rcll}
{\it GE}[U](M)&=& 1.25 \\
{\it CC}(M) & = & \log 3\approx 1.5849625 
\end{array}
\]
We now state the main results of the section, which show that the
bounding problems for ${\it SE}[U]$, ${\it ME}[U]$, ${\it GE}[U]$, and
${\it CC}$ are PP-hard.
\begin{theorem}
$\text{PP}\subseteq B_{\it SE}[U]$
\label{thm:ppse}
\end{theorem}
\begin{theorem}
$\text{PP}\subseteq B_{\it ME}[U]$
\label{thm:ppme}
\end{theorem}
\begin{theorem}
$\text{PP}\subseteq B_{\it GE}[U]$
\label{thm:ppge}
\end{theorem}
\begin{theorem}
$\text{PP}\subseteq B_{\it CC}$
\label{thm:ppcc}
\end{theorem}

We remind that the above results hold (even) when the bounding
problems $B_{\it SE}[U]$, $B_{\it ME}[U]$, $B_{\it GE}[U]$, and
$B_{\it CC}$ are restricted to loop-free boolean programs.  We also
note that the results hold even when the programs are restricted to
those without low security inputs.  These results are proven by a
reduction from MAJSAT, which is a PP-complete problem.  MAJSAT is the
problem of deciding, given a boolean formula $\phi$ over variables
$\vect x$, if there are more than $2^{|\vect x|-1}$ satisfying
assignments to $\phi$ (i.e., whether the majority of the assignments
to $\phi$ are satisfying).

\subsection{Complexities for Belief and Channel Capacity Like Definitions}

\label{sec:complexbeliefcclike}

This section investigates the complexity theoretic hardness of the
bounding problems for the belief-based definition and the
channel-capacity-like definition of quantitative information flow
introduced in Section~\ref{sec:ksafetybelief} and
Section~\ref{sec:ksafetycclike}.  As in Section~\ref{sec:complex}, we
focus on loop-free boolean programs.

Below shows the complexity results for the belief-based comparison
problems for loop-free boolean
programs~\cite{yasuoka:toplas2010submit}.
\begin{itemize}
\item[(1)] Checking ${\it BE}[\aseq{U,h,\ell}](M_1)
  \leq {\it BE}[\aseq{U,h,\ell}](M_2)$ is PP-hard.
\item[(2)] Checking $\forall h,\ell.{\it BE}[\aseq{U,h,\ell}](M_1)
  \leq {\it BE}[\aseq{U,h,\ell}](M_2)$ is PP-hard.
\item[(3)] Checking $\forall \mu,h,\ell.{\it BE}[\aseq{\mu,h,\ell}](M_1)
  \leq {\it BE}[\aseq{\mu,h,\ell}](M_2)$ is coNP-complete.
\end{itemize}

First, we prove that the two types of bounding problems for the
belief-based definition, $B_{\it BE1}$ and $B_{\it BE2}$, are both
PP-hard.
\begin{theorem}
\label{thm:ppbe1}
$\text{PP}\subseteq B_{\it BE1}[\aseq{U,h,\ell}]$
\end{theorem}
\begin{theorem}
\label{thm:ppbe2}
$\text{PP}\subseteq B_{\it BE2}[U]$
\end{theorem}
As in Section~\ref{sec:complex}, the above theorems are proven by
a reduction from MAJSAT.  They show that the bounding problems for
${\it BE}[U]$ are complexity theoretically difficult.

Next, we prove the hardness results for the channel-capacity like
definitions of quantitative information flow.
Theorems~\ref{thm:ppsecc} and \ref{thm:ppmecc} for $B_{\it SECC}$ and
$B_{\it MECC}$ follow from the equivalence $\max_\mu {\it SE}[\mu](M)
= \max_\mu {\it ME}[\mu](M) = {\it CC}(M)$
(cf. Section~\ref{sec:ksafetycclike}) and Theorem~\ref{thm:ppcc}.
Theorem~\ref{thm:ppgecc} for $B_{\it GECC}$ follows from
Theorem~\ref{thm:ppge} and the equivalence $\max_\mu {\it GE}[\mu](M)
= \max_\ell {\it GE}[U\otimes\dot{\ell}](M)$
(cf. Lemma~\ref{lem:gecc}).
\begin{theorem}
\label{thm:ppsecc}
  $\text{PP}\subseteq B_{\it SECC}$
\end{theorem}
\begin{theorem}
\label{thm:ppmecc}
  $\text{PP}\subseteq B_{\it MECC}$
\end{theorem}
\begin{theorem}
\label{thm:ppgecc}
  $\text{PP}\subseteq B_{\it GECC}$
\end{theorem}

\begin{sloppypar}
Finally, the following coNP-completeness results for $B_{\it
  BE1CC}[h,\ell]$ and $B_{\it BE2CC}$ follow from their equivalent to
non-interference and the fact that checking non-interference is
coNP-complete for loop-free boolean programs
(cf. Section~\ref{sec:complex}).
\end{sloppypar}
\begin{theorem}
\label{thm:conpbe3}
$B_{\it BE1CC}[h,\ell]$ is coNP-complete.
\end{theorem}

\begin{theorem}
\label{thm:conpbe4}
$B_{\it BE2CC}$ is coNP-complete.
\end{theorem}

\section{Discussion}

\label{sec:discussion}

\subsection{Bounding the Domains}

The notion of $k$-safety property, like the notion of safety property
from where it extends, is defined over all programs regardless of
their size.  (For example, non-interference is a $2$-safety property
for all programs and unreachability is a safety property for all
programs.)  But, it is easy to show that the bounding problems would
become ``$k$-safety'' properties if we constrained and bounded the
input domains because then the size of the semantics (i.e., the
input/output pairs) of such programs would be bounded by
$|\mathbb{H}|\hspace{-0.2em}\times\hspace{-0.2em}|\mathbb{L}|$.  In
this case, the problems are at most
$|\mathbb{H}|\hspace{-0.2em}\times\hspace{-0.2em}|\mathbb{L}|$-safety.
(And the complexity theoretic hardness degenerates to a constant.)
But, like the $k$-safety bounds obtained by fixing $q$ constant
(cf. Section~\ref{sec:ksafetyconst}), these bounds are high for all
but very small domains and are unlikely to lead to a practical
verification method.  Also, because a bound on the high security input
domain puts a bound on the maximum information flow, the bounding
problems become a tautology for $q \geq c$, where $c$ is the maximum
information flow for the respective definition.

\subsection{Low Security Inputs}

\label{sec:lowsecinputs}

Recall the results from Section~\ref{sec:ksafetyconst} that, under a
constant bound, the bounding problems for both the min-entropy based
definition and the guessing-entropy based definition with the uniform
distribution are $k$-safety for programs without low security inputs,
but not for those with.  The reason for the non-$k$-safety results is
that the definitions of quantitative information flow ${\it ME}$ and
${\it GE}$ (and in fact, also ${\it SE}$) use the conditional entropy
over the low security input distribution and are parameterized by the
distribution.  This means that the quantitative information flow of a
program is averaged over the low security inputs according to the
distribution.  Therefore, by arbitrarily increasing the number of low
security inputs, given any set of traces $T$, it becomes possible to
find a program containing $T$ whose information flow is arbitrarily
close to $0$ (at least under the uniform distribution).  This appears
to be a property intrinsic to any definition of quantitative
information flow defined via conditional entropy over the low security
inputs and is parameterized by the distribution of low security
inputs.  Note that the channel-capacity-like definitions do not share
this property as it is defined to be the maximum over the
distributions.  The non-$k$-safety result for $B_{\it SE}[U]$ holds
even in the absence of low security inputs because the Shannon entropy
of a program is the average of the {\em
  surprisal}~\cite{clarkson:csf2005} of the individual observations,
and so by increasing the number of high security inputs, given any set
of traces $T$, it becomes possible to find a program containing $T$
whose information flow is arbitrarily close to $0$.  The
non-$k$-safety results for $B_{\it BE1}[\aseq{U,h}]$ and $B_{\it
  BE2}[U]$ hold for similar reasons.\footnote{They are, respectively,
  the surprisal of a particular input, and the maximum surprisal over
  all the inputs.}

\section{Related Work}

\label{sec:related}

This work continues our recent
research~\cite{DBLP:conf/csfw/yasuoka2010} on investigating the
hardness and possibilities of verifying quantitative information flow
according to the formal definitions proposed in
literature~\cite{clarkson:csf2005,denning82,clarkjcs2007,malacaria:popl2007,smith09,kopf07,DBLP:conf/sp/BackesKR09,mccamant:pldi2008,malacaria08,NMS2009,Braun:09:MFPS,DBLP:conf/csfw/KopfS10}.
Much of the previous research has focused on information theoretic
properties of the definitions and proposed approximate (i.e.,
incomplete and/or unsound) methods for checking and inferring
quantitative information flow according to such definitions.  In
contrast, this paper (along with our recent
paper~\cite{DBLP:conf/csfw/yasuoka2010}) investigates the hardness and
possibilities of precisely checking and inferring quantitative
information flow according to the definitions.

This paper has shown that the bounding problem, that is, the problem
of checking $\mathcal{X}(M) \leq q$ given a program $M$ and a positive
real $q$, is quite hard (for various quantitative information flow
definitions $\mathcal{X}$).  This is in contrast to our previous paper
that has investigated the hardness and possibilities of the comparison
problem, that is, the problem of checking $\mathcal{X}(M_1) \leq
\mathcal{X}(M_2)$ given programs $M_1$ and $M_2$.  To the best of our
knowledge, this paper is the first to investigate the hardness of the
bounding problems.  But, the hardness of quantitative information flow
inference, a harder problem, follows from the results of our previous
paper, and Backes et al.~\cite{DBLP:conf/sp/BackesKR09} and also
Heusser and Malacaria~\cite{DBLP:conf/ifip1-7/HeusserM09} have
proposed a precise inference method that utilizes self composition and
counting algorithms.  Also, independently from our work, Heusser and
Malacaria~\cite{DBLP:conf/acsac/Heusser2010} have recently applied the
self-composition method outlined in Section~\ref{sec:ksafetyconst} for
checking the channel-capacity-based quantitative information flow.

\section{Conclusion}

\label{sec:concl}

In this paper, we have formalized and proved the hardness of the
bounding problem of quantitative information flow, which is a form of
(precise) checking problem of quantitative information flow.  We have
shown that no bounding problem is a $k$-safety property for any $k$,
and therefore that it is not possible to reduce the problem to a
safety problem via self composition, at least when the quantity to
check against is unrestricted.  The result is in contrast to
non-interference and the quantitative information flow comparison
problem with universally quantified distribution, which are $2$-safety
properties.  We have also shown a complexity theoretic gap with these
problems, which are coNP-complete, by proving the PP-hardness of the
bounding problems, when restricted to loop-free boolean programs.

We have also shown that the bounding problems for some quantitative
information flow definitions become $k$-safety for different $k$'s
under certain conditions when the quantity to check against is
restricted to be a constant, highlighting interesting disparities
among the different definitions of quantitative information flow.

It is interesting to note that, as with the comparison problems, the
bounding problems become comparatively easier when the input
distribution becomes universally quantified.  That is, as our previous
work~\cite{DBLP:conf/csfw/yasuoka2010} has shown that checking if
$\forall \mu.\mathcal{Y}[\mu](M_1) \leq \mathcal{Y}[\mu](M_2)$ is
often easier than checking if $\mathcal{Y}[U](M_1) \leq
\mathcal{Y}[U](M_2)$ (for various quantitative information flow
definitions $\mathcal{Y}$), we have shown that the problem of checking
$\forall \mu.\mathcal{Y}[\mu](M) \leq q$ is often easier than the
problem of checking $\mathcal{Y}[U](M) \leq q$.

\section*{Acknowledgments}
  This work was supported by MEXT KAKENHI 23700026, 22300005, and the
  Global COE Program ``CERIES.''

\bibliographystyle{abbrv}
\bibliography{boundflow}

\appendix
\section{Proofs}
\label{appendix}

We define some abbreviations.
\begin{definition}
\label{def:distabrv}
  $\mu(x)\triangleq \mu(X=x)$
\end{definition}
We use the above notation whenever the correspondences between random variables
and their values are clear.

We define some useful abbreviations for programs having low security inputs.

\begin{definition}
  $M[\mathbb{H},\ell]=\aset{o\mid \exists h\in\mathbb{H}.
    o=M(h,\ell)}$
\end{definition}

\begin{definition}
$M(\ell)=\lambda h. M(h,\ell)$
\end{definition}

Note that $M(\ell)$ is the program $M$ restricted to the low security
input $\ell$, and that $M[\mathbb{H},\ell]$ is the set of outputs of
$M(\ell)$.

We elide the parameter $q$ from the input to the bounding problems when
it is clear from the context (e.g., when $q$ is held constant).  For
example, we write $B_{\it SE}[U](M)$ and $M \in B_{\it SE}[U]$ instead
of $B_{\it SE}[U](M,q)$ or $(M,q) \in B_{\it SE}[U]$.

We note the following properties of deterministic
programs~\cite{clark05}.
\begin{lemma}
\label{lem:detse}
Let $M$ be a program without low-security inputs, $M'$ be a program
with low-security inputs.  Then, we have
${\it SE}[\mu](M)=\mathcal{I}[\mu](O;H)=\mathcal{H}[\mu](O)$
and
${\it SE}[\mu](M')=\mathcal{I}[\mu](O;H|L) = \mathcal{H}[\mu](O|L)$
\end{lemma}

\begin{definition}
\[
\begin{array}{l}
In(\mu,X,x)=|\aset{x'\in X\mid \mu(x')\ge\mu(x)}|
\end{array}
\]
\end{definition}
Intuitively, $In(\mu,X,x)$ is the order of $x$ defined in terms of
$\mu$.

\begin{lemma}
\[
{\mathcal G}[\mu](X)=\Sigma_{1\le i\le |X|}i\mu(x_i)=\Sigma_{x\in X}In(\mu,X,x)\mu(x)
\]
\end{lemma}
\begin{proof}
Trivial.
\end{proof}

\begin{reflemma}{\ref{lem:be}}
  Let $\mu$ be a belief, $h_\mathcal{E}$ be a high-security input,
  $\ell_\mathcal{E}$ be a low-security input.  Then, ${\it
    BE}[\langle
    \mu,h_\mathcal{E},\ell_\mathcal{E}\rangle](M)=-\log\Sigma_{h\in
    \aset{h'\mid
      M(h',\ell_\mathcal{E})=M(h_\mathcal{E},\ell_\mathcal{E})}}
  \mu(h)$.
\end{reflemma}
\begin{proof}
By definition, we have
\[
\begin{array}{l}
  {\it BE}[\langle
    \mu,h_\mathcal{E},\ell_\mathcal{E}\rangle](M) \\
  \qquad =D(\mu\rightarrow \dot{h_\mathcal{E}})-D(\mu |
  o_\mathcal{E}\rightarrow
  \dot{h_\mathcal{E}})\\
  \qquad =\sum_h
  \dot{h_\mathcal{E}}(h)\log\frac{\dot{h_\mathcal{E}}(h)}{\mu(h)} -
  \sum_h \dot{h_\mathcal{E}}(h)\log\frac{\dot{h_\mathcal{E}}(h)}{\mu |
    o_\mathcal{E}(h)}\\
  \qquad =\log\frac{1}{\mu(h_\mathcal{E})} +\log\frac{\mu(h_\mathcal{E})}{\sum_{h\in \aset{h'\mid M(h',\ell_\mathcal{E}) =
    M(h_\mathcal{E},\ell_\mathcal{E})}} \mu(h)}\\
\qquad=-\log\sum_{h\in \aset{h'\mid M(h',\ell_\mathcal{E}) =
    M(h_\mathcal{E},\ell_\mathcal{E})}} \mu(h)

\end{array}
\]

\end{proof}

\begin{reftheorem}{\ref{thm:nonint}}
Let $\mu$ be a distribution such that $\forall
h\in\mathbb{H},\ell\in\mathbb{L}.\mu(h,\ell)>0$.  Then,
\begin{itemize}
\item $M$ is non-interferent if and only if ${\it SE}[\mu](M) \leq 0$.
\item $M$ is non-interferent if and only if ${\it ME}[\mu](M) \leq 0$.
\item $M$ is non-interferent if and only if ${\it GE}[\mu](M) \leq 0$.
\item $M$ is non-interferent if and only if ${\it BE}[\aseq{\mu',h,\ell}](M)\le 0$.\footnote{Recall Definition~\ref{def:beliefqif} that $\mu'$ is a
distribution over $\mathbb{H}$ such that $\mu'(h) > 0$ for all $h \in
\mathbb{H}$.}
\item $M$ is non-interferent if and only if ${\it CC}(M) \leq 0$.  
\end{itemize}
\end{reftheorem}
\begin{proof}
  Let $\mathbb{O}=\aset{M(h,\ell)\mid h\in\mathbb{H}\wedge
  \ell\in\mathbb{L}}$.
\begin{itemize}
\item ${\it SE}$

\hspace{0.5em}
  (See~\cite{clark05}.)
\vspace{0.5em}

\item ${\it ME}$
\begin{itemize}
\item $\Rightarrow$

  Suppose $M$ is non-interferent.  By the definition, it suffices to
  show that
\[
\mathcal{V}[\mu](H|L)=\mathcal{V}[\mu](H|L,O)
\]
That is,
\[
\sum_\ell\mu(\ell)\max_h \mu(h|\ell) =
\sum_{\ell,o}\mu(\ell,o)\max_h\mu(h|\ell,o)
\]
We have for any $\ell_x$ and $o_x$ such that $\mu(\ell_x,o_x)>0$,
$\mu(\ell_x,o_x)=\mu(\ell_x)$, and 
for all $h_y$, $\ell_y$, and $o_y$ such that
$\mu(h_y,\ell_y,o_y)>0$, for any $h'_y$ and
$o'\in\mathbb{O}\setminus\aset{o_y}$, $\mu(h'_y,\ell_y,o'_y)=0$.
 Therefore, we
have
\[
\begin{array}{rcl}
  \sum_{\ell,o}\mu(\ell,o)\max_h\mu(h|\ell,o)&=&\sum_{\ell,o}\mu(\ell,o)\max_h\frac{\mu(h,\ell,o)}{\mu(\ell,o)}\\
&=&\sum_{\ell}\mu(\ell)\max_h\mu(h|\ell)
\end{array}
\]
\item $\Leftarrow$

  We prove the contraposition.  Suppose $M$ is interferent.  That is,
  there exist $h_1$, $h_2$, and $\ell'$ such that
  $M(h_1,\ell')\not=M(h_2,\ell')$.  Let $o_1=M(h_1,\ell')$ and
  $o_2=M(h_2,\ell')$.  We have
\[
  \sum_{\ell}\mu(\ell)\max_h\mu(h|\ell)=A+\max_h\mu(h,\ell')
\]
where
$A=\sum_{\ell\in\mathbb{L}\setminus\aset{\ell'}}\max_h\mu(h,\ell)$.
And,
\[
  \sum_{\ell,o}\mu(\ell,o)\max_h\mu(h|\ell,o)=B+\sum_o\max_h\mu(h,\ell',o)
\]
where
$B=\sum_{(\ell,o)\in(\mathbb{L}\setminus\aset{\ell'})\times\mathbb{O}}\max_h\mu(h,\ell,o)$.   
Trivially, we have $A\le B$ and
\[
\max_h\mu(h,\ell')<\sum_o\max_h\mu(h,\ell',o)
\]
Therefore, we have ${\it ME}[\mu](M)>0$.
\end{itemize}

\item ${\it GE}$
\begin{itemize}
\item $\Rightarrow$

  Suppose $M$ is non-interferent.  By the definition,
\[
\begin{array}{l}
  {\it GE}[\mu](M)\\
  \quad=\sum_\ell\sum_h In(\lambda h'.\mu(h',\ell),\mathbb{H},h)\mu(h,\ell)\\
  \qquad-\sum_{\ell,o}\sum_h In(\lambda h'.\mu(h',\ell,o),\mathbb{H},h)\mu(h,\ell,o)\\
\quad=\sum_\ell\sum_h In(\lambda h'.\mu(h',\ell),\mathbb{H},h)\mu(h,\ell)\\
\qquad-\sum_\ell\sum_h In(\lambda h'.\mu(h',\ell),\mathbb{H},h)\mu(h,\ell)\\
\quad=0
\end{array}
\]
since for all $h_x$, $\ell_x$, and $o_x$ such that
$\mu(h_x,\ell_x,o_x)>0$, for any $h'_x$ and
$o'\in\mathbb{O}\setminus\aset{o_x}$, $\mu(h'_x,\ell_x,o'_x)=0$.
\item $\Leftarrow$

  We prove the contraposition.  Suppose $M$ is interferent.  That is,
  there exist $h_1$, $h_2$, and $\ell'$ such that
  $M(h_1,\ell')\not=M(h_2,\ell')$.  Let $o_1=M(h_1,\ell')$ and
  $o_2=M(h_2,\ell')$.  By the definition,
\[
\begin{array}{l}
  {\it GE}[\mu](M)\\
  \quad=\sum_\ell\sum_h In(\lambda h'.\mu(h',\ell),\mathbb{H},h)\mu(h,\ell)\\
  \qquad-\sum_{\ell,o}\sum_h In(\lambda h'.\mu(h',\ell,o),\mathbb{H},h)\mu(h,\ell,o)\\
\quad=A+\sum_h In(\lambda h'.\mu(h',\ell'),\mathbb{H},h)\mu(h,\ell')\\
\qquad-B-\sum_o\sum_h In(\lambda h'.\mu(h',\ell',o),\mathbb{H},h)\mu(h,\ell',o)
\end{array}
\]
where 
\[
\begin{array}{l}
A=\sum_{\ell\in\mathbb{L}\setminus\aset{\ell'}}\sum_h
In(\lambda h'.\mu(h',\ell'),\mathbb{H},h)\mu(h,\ell')\\
B=\sum_{(\ell,o)\in(\mathbb{L}\setminus\aset{\ell'})\times\mathbb{O}}\sum_h
In(\lambda h'.\mu(h',\ell',o),\mathbb{H},h)\mu(h,\ell',o)
\end{array}
\]
Trivially, we have $A\ge B$ and
\[
\begin{array}{l}
\sum_h In(\lambda h'.\mu(h',\ell'),\mathbb{H},h)\mu(h,\ell')\\
\qquad>\sum_o\sum_h In(\lambda h'.\mu(h',\ell',o),\mathbb{H},h)\mu(h,\ell',o)
\end{array}
\]
Therefore, we have ${\it GE}[\mu](M)>0$.
\end{itemize}

\item ${\it BE}$
\begin{itemize}
\item $\Rightarrow$

Suppose $M$ is non-interferent.  By Lemma~\ref{lem:be}, for any
$\mu$, $h$, and $\ell$,
\[
  {\it BE}[\aseq{\mu,h,\ell}](M)=-\log\Sigma_{h'\in
    \aset{h''\mid
      M(h'',\ell)=M(h,\ell)}}\mu(h')=0
\]
\item $\Leftarrow$

  We prove the contraposition. Suppose $M$ is interferent.  That is,
  there exist $h_1$, $h_2$, and $\ell'$ such that
  $M(h_1,\ell')\not=M(h_2,\ell')$.  Let $\mu'$ be a distribution such
  that for any $h'$, $\mu'(h')>0$.  Then, by Lemma~\ref{lem:be}, we
  have for any $h$,
\[
  {\it BE}[\aseq{\mu',h,\ell'}](M) =-\log\Sigma_{h'\in
    \aset{h''\mid
      M(h'',\ell')=M(h,\ell')}}\mu'(h')
  >0
\]
\end{itemize}

\item ${\it CC}$
\begin{itemize}
\item $\Rightarrow$

  Suppose $M$ is non-interferent.  By Lemma~\ref{lem:detse}, for any
  $\mu$,
\[
\begin{array}{rcl}
  {\it SE}[\mu](M)&=&\mathcal{H}[\mu](O|L)\\
  &=&\sum_o\sum_\ell \mu(o,\ell)\log\frac{\mu(\ell)}{\mu(o,\ell)}\\
  &=&0
\end{array}
\]
since $\mu(o,\ell)=\mu(\ell)$.  Therefore, we have $\forall\mu.{\it
  SE}[\mu](M)=0$.  It follows that ${\it CC}(M)=0$.
\item $\Leftarrow$

  We prove the contraposition. Suppose $M$ is interferent.  That is, there
  exist $h_1$, $h_2$, and $\ell'$ such that
  $M(h_1,\ell')\not=M(h_2,\ell')$.  Let $o_1=M(h_1,\ell')$, and
  $o_2=M(h_2,\ell')$.  Then, there exist $\mu'$ such that 
\[
\begin{array}{rcl}
  {\it SE}[\mu'](M)&=&\mathcal{H}[\mu'](O|L)\\
  &\ge&\mu'(o_1,\ell')\log\frac{\mu'(\ell')}{\mu'(o_1,\ell')}+\mu'(o_2,\ell')\log\frac{\mu'(\ell')}{\mu'(o_2,\ell')}\\
  &>&0
\end{array}
\]
And, we have ${\it SE}[\mu'](M)\le {\it CC}(M)$.
\end{itemize}

\end{itemize}
\end{proof}

We note the following equivalence of {\it CC} and {\it ME}[U] for
programs without low security inputs~\cite{smith09}.
\begin{lemma}
\label{lem:ccme}
Let $M$ be a program without low security input.  Then, 
${\it CC}(M) = {\it  ME}[U](M)$.

\end{lemma}

\begin{reftheorem}{\ref{thm:senk2}}
Neither $B_{\it SE}[U]$, $B_{\it ME}[U]$, $B_{\it GE}[U]$, nor $B_{\it CC}$ is a
k-safety property for any k such that $k > 0$.
\end{reftheorem}
\begin{proof}
\noindent
\begin{itemize}
\item  $B_{\it SE}[U]$ is not a k-safety problem for any k such that
  $k>0$.

Trivial by Theorem~\ref{thm:senk}.

\item  $B_{\it ME}[U]$ is not a k-safety property for any k such that
  $k>0$.

Trivial by Theorem~\ref{thm:mek}.

\item $B_{\it GE}[U]$ is not a k-safety property for any k such that
  $k>0$.

Trivial by Theorem~\ref{thm:gek}.

\item $B_{\it CC}$ is not a k-safety property for any k such that $k>0$.

  Trivial from Lemma~\ref{lem:ccme} and the fact that $B_{\it ME}[U]$
  is not a k-safety property for any k.

\end{itemize}
\end{proof}

Malacaria and Chen~\cite{malacaria08} have proved the following result
relating the channel-capacity based quantitative information flow with
the number of outputs.  
\begin{lemma}
\label{lem:ccloglow}
Let $M$ be a program (with low security input).  Then,
\[
\begin{array}{l}
{\it CC}(M) = 
\max_{\ell\in\mathbb{L}}\log |M[\mathbb{H},\ell]|
\end{array}
\]
\end{lemma}

\begin{reftheorem}{\ref{thm:cck}}
  Let $q$ be a constant.  Then, $B_{\it CC}$ is
  $\lfloor2^q\rfloor+1$-safety, but it is not $k$-safety for any $k
  \leq \lfloor2^q\rfloor$.
\end{reftheorem}
\begin{proof}
  We prove that $B_{\it CC}$ is $\lfloor 2^q\rfloor +1$-safety.  Let
  $M$ be a program such that $M\not\in B_{\it CC}$.  By
  Lemma~\ref{lem:ccloglow}, it must be the case that there exists
  $\ell$ such that $|M[\mathbb{H},\ell]|\ge\lfloor 2^q\rfloor +1$.
  Then, there exists $T \subseteq \sembrack{M}$ such that $|T| \leq
  \lfloor2^q\rfloor+1$, $\ran(T) \geq \lfloor2^q\rfloor+1$, and for
  all $((h,\ell'),o) \in T$, $\ell' = \ell$.  Then, by
  Lemma~\ref{lem:ccloglow}, it follows that for any program $M'$ such
  that $T \subseteq \sembrack{M'}$, $M' \not\in B_{\it CC}$.
  Therefore, $B_{\it CC}$ is a $\lfloor 2^q\rfloor +1$-safety
  property.

  Finally, we prove that $B_{\it CC}[U]$ is not $k$-safety for any $k
  \le \lfloor 2^q\rfloor$.  Let $k \leq \lfloor2^q\rfloor$. For a
  contradiction, suppose $B_{\it CC}$ is a $k$-safety property.  Let
  $M$ be a program such that $M\not\in B_{\it CC}$.  Then, there
  exists $T$ such that $|T|\le k$ and $T\subseteq\sembrack{M}$, and
  for any $M'$ such that $T\subseteq \sembrack{M'}$, $(M',q)\not\in
  B_{\it CC}$.  Let $T=\aset{(h_1,o_1),\dots,(h_i,o_i)}$.  Let
  $\bar{M}$ be a program such that $\sembrack{\bar{M}} = T$.  More
  formally, let $\bar{M}$ be the following program.
\[
\bar{M}(h_1)=o_1, \bar{M}(h_2)=o_2, \dots, \bar{M}(h_i)=o_i
\]
Then, we have
\[
{\it CC}(\bar{M})=\log |\aset{o_1,o_2,\dots,o_i}|\le \log k\le q
\]
It follows that $(\bar{M},q)\in {\it CC}$, but
$T\subseteq\sembrack{\bar{M}}$.  Therefore, this leads to a
contradiction.
\end{proof}

\begin{reftheorem}{\ref{thm:mek}}
  Let $q$ be a constant, and suppose $B_{\it ME}[U]$ only takes
  programs without low security inputs.  Then, $B_{\it ME}[U]$ is
  $\lfloor 2^q\rfloor +1$-safety, but it is not $k$-safety for any $k
  \leq \lfloor2^q\rfloor$.
\end{reftheorem}
\begin{proof}
  Straightforward by Theorem~\ref{thm:cck} and Lemma~\ref{lem:ccme}.
\end{proof}

\begin{lemma}
  Let $M$ be a program without low security inputs.  Then, we have
  ${\it GE}[U](M)=\frac{n}{2}-\frac{1}{2n}\sum_o|\mathbb{H}_o|^2$
  where $n$ is the number of inputs, and $\mathbb{H}_o=\aset{h\mid
    o=M(h)}$.
\label{lem:geu}
\end{lemma}
\begin{proof}
By the definition, we have
\[
\begin{array}{rcl}
  {\it GE}[U](M)&=&\mathcal{G}[U](H)-\mathcal{G}[U](H|O)\\
  &=&\sum_h In(U,\mathbb{H},h)U(h)\\
&&\qquad-\sum_o U(o)\sum_h In(\lambda h'.U(h'|o),\mathbb{H}_o,h) U(h|o)\\
  &=&\frac{1}{n}\frac{1}{2}n(n+1)-\sum_o \frac{|\mathbb{H}_o|}{n}\frac{1}{2}\frac{1}{|\mathbb{H}_o|}|\mathbb{H}_o|(|\mathbb{H}_o|+1)\\
  &=&\frac{n}{2}-\frac{1}{2n}\sum_o |\mathbb{H}_o|^2
\end{array}
\]
\end{proof}

\begin{lemma}
  Let $M$ and $M'$ be low-security input free programs such that
  $\sembrack{M'}=\sembrack{M}\cup\aset{(h,o)}$ and
  $h\not\in\dom(\sembrack{M})$.  Then, we have ${\it GE}[U](M)\le{\it
  GE}[U](M')$.
\label{lem:gemono2}
\end{lemma}
\begin{proof}
  We prove ${\it GE}[U](M')-{\it GE}[U](M)\ge 0$.  Let $n =
  |\sembrack{M}|$, $\mathbb{O}=\ran(\sembrack{M})$, $\mathbb{H} =
  \dom(M)$, and $\mathbb{H}_o = \aset{h\in\mathbb{H}\mid o= M(h)}$.

  By Lemma~\ref{lem:geu}, we have
\[
\begin{array}{l}
  {\it GE}[U](M')-{\it GE}[U](M)\\
\hspace{4em}=\frac{n+1}{2}-\frac{1}{2(n+1)}(B + (|\mathbb{H}_{o}|+1)^2)
  -\frac{n}{2}+\frac{1}{2n}(B+ |\mathbb{H}_{o}|^2)\\
\hspace{4em}  =\frac{1}{2n(n+1)}((n-|\mathbb{H}_{o}|)^2+B)\ge 0
\end{array}
\]
where $B=\sum_{o'\in \mathbb O\setminus \aset{o}} |\mathbb{H}_{o'}|^2$
 and $\mathbb{H}_{o'}=\aset{h\mid o'=M(h)}$.
\end{proof}

\begin{lemma}
Let $q\ge\frac{1}{2}$.  Let $M$ be a program without low security
inputs such that ${\it GE}[U](M)>q$ and $\forall
M'.\sembrack{M'}\subsetneq\sembrack{M}\Rightarrow {\it GE}[U](M')\le
q$.  Then, it must be the case that $|\sembrack{M}|\le\lfloor
\frac{(\lfloor q\rfloor +1)^2}{\lfloor q\rfloor +1 -q}\rfloor +1$.
\label{lem:gemax}
\end{lemma}
\begin{proof}
  Let $n$ be the integer such that $n=|\sembrack{M}|$.  If $M$ returns
  only one output, we have ${\it GE}[U](M)=0$.  Therefore, $M$ must
  have more than 1 output as ${\it GE}[U](M)>q$.  By
  Lemma~\ref{lem:geu}, we have for any $o'$
\[
\begin{array}{rcl}
{\it GE}[U](M)&=&\frac{n}{2}-\frac{1}{2n}(B+(n-i)^2)\\
&=&i-\frac{1}{2n}(B+i^2)
\end{array}
\]
where $i=\sum_{o\in\mathbb{O}\setminus\aset{o'}}|\mathbb{H}_o|$ and
$B=\sum_{o\in\mathbb{O}\setminus\aset{o'}}|\mathbb{H}_o|^2$.  Because
${\it GE}[U](M)>q$, we have $i>q$.  Then, we have
\[
\begin{array}{rcl}
{\it GE}[U](M)> q&\;\textrm{iff}\;&i-\frac{B+i^2}{2n}> q\\
&\;\textrm{iff}\;&n>\frac{B+i^2}{2(i-q)}
\end{array}
\]
By the definition of $M$, we have $\forall
M'.\sembrack{M'}\subsetneq\sembrack{M}\Rightarrow {\it GE}[U](M')\le
q$.  Let $\sembrack{\bar{M}}=\sembrack{M}\setminus \aset{(h',o')}$
where $M(h')=o'$.  Then, we have
\[
\begin{array}{rcl}
{\it GE}[U](\bar{M})\le q\;&\textrm{iff}&\;i-\frac{B+i^2}{2(n-1)}\le q\\
&\textrm{iff}&\;n\le\frac{B+i^2}{2(i-q)}+1
\end{array}
\]
Hence, we have
\[
\frac{B+i^2}{2(i-q)}< n\le\frac{B+i^2}{2(i-q)}+1
\]
Because $B=\sum_{o\in\mathbb{O}\setminus\aset{o'}}|\mathbb{H}_o|^2$
and $i=\sum_{o\in\mathbb{O}\setminus\aset{o'}}|\mathbb{H}_o|$, the
largest $n$ occurs when $B=i^2$.  That is, when $M$ has exactly two
outputs.  Therefore, it suffices to prove the lemma for just such
$M$'s.

Now, we prove $|\sembrack{M}|\le\lfloor \frac{(\lfloor q\rfloor
  +1)^2}{\lfloor q\rfloor +1 -q}\rfloor +1$.  Recall that
$i=\sum_{o\in\mathbb{O}\setminus\aset{o'}}|\mathbb{H}_o|$.  Let
$j=n-i$.  We have
\[
\begin{array}{rcl}
{\it GE}[U](M)&=&i-\frac{1}{2n}(i^2+i^2)\\
&=&j-\frac{j^2}{n}\\
&>&q
\end{array}
\]
This means that $j>q$.  Recall that
$\sembrack{\bar{M}}=\sembrack{M}\setminus \aset{(h',o')}$ where
$M(h')=o'$.  Then, we have
\[
\begin{array}{rcl}
{\it GE}[U](\bar{M})\le q\;&\textrm{iff}&\;i-\frac{i^2}{n-1}\le q\\
&\textrm{iff}&\;n\le\frac{i^2}{i-q}+1
\end{array}
\]
Because $n$ is an integer, we have $n\le\lfloor\frac{i^2}{i-q}\rfloor
+1$ and $n\le\lfloor\frac{j^2}{j-q}\rfloor +1$.  Let
$f=\frac{i^2}{i-q}+1=\frac{j^2}{j-q}+1$.  By elementary real analysis,
it can be shown that for integers $i$ and $j$ such that $i>q$ and
$j>q$, $f$ attains its maximum value when $i=\lfloor q\rfloor +1$ or
$j=\lfloor q\rfloor +1$.  Therefore, it follows that
$|\sembrack{M}|=n\le\lfloor\frac{(\lfloor q\rfloor +1)^2}{\lfloor
  q\rfloor +1 -q}\rfloor +1$.
\end{proof}

\begin{lemma}
  Let $q\ge\frac{1}{2}$.  Let $M$ be a program without low-security
  inputs such that ${\it GE}[U](M)>q$.  Then, there exists $T$ such
  that
\begin{itemize}
\item $T\subseteq \sembrack{M}$
\item $|T|\le \lfloor\frac{(\lfloor q\rfloor +1)^2}{\lfloor q\rfloor
    +1 -q} \rfloor +1$
\item ${\it GE}[U](M')>q$ where $\sembrack{M'}=T$.
\end{itemize}
\label{lem:gemax3}
\end{lemma}
\begin{proof}
  Let $q\ge\frac{1}{2}$.  Let $M$ be a program such that ${\it
    GE}[U](M)>q$.  By Lemma~\ref{lem:gemono2} and the fact that ${\it
    GE}[U](M)$ is bounded by $\frac{|\sembrack{M}|}{2}$, there exists
  $T$ such that
\begin{itemize}
\item $T\subseteq \sembrack{M}$
\item ${\it GE}[U](M')>q$ where $\sembrack{M'}=T$
\item $\forall T'\subseteq T. {\it GE}[U](\bar{M})\le q$ where
  $\sembrack{\bar{M}}=T'$.
\end{itemize}
By Lemma~\ref{lem:gemax}, we have $|T|\le \lfloor\frac{(\lfloor
  q\rfloor +1)^2}{\lfloor q\rfloor +1 -q} \rfloor +1$.  Therefore, we
have the conclusion.
\end{proof}

\begin{reftheorem}{\ref{thm:gek}}
  Let $q$ be a constant, and suppose $B_{\it GE}[U]$ only takes
  programs without low security inputs.  If $q\ge\frac{1}{2}$, then,
  $B_{\it GE}[U]$ is $\lfloor\frac{(\lfloor q\rfloor +1)^2}{\lfloor
    q\rfloor +1 -q} \rfloor +1$-safety, but it is not $k$-safety for
  any $k \leq \lfloor\frac{(\lfloor q\rfloor +1)^2}{\lfloor q\rfloor
    +1 -q} \rfloor$.  Otherwise, $q<\frac{1}{2}$ and $B_{\it GE}[U]$
  is $2$-safety, but it is not $1$-safety.
\end{reftheorem}
\begin{proof}
  First, we prove that $B_{\it GE}[U]$ for programs without
  low-security inputs is $\lfloor\frac{(\lfloor q\rfloor
    +1)^2}{\lfloor q\rfloor +1 -q} \rfloor +1$-safety for $q \geq
  \frac{1}{2}$.  By the definition of $k$-safety, for any $M$ such
  that $M\not\in B_{\it GE}[U]$, there exists $T$ such that
\begin{enumerate}
\item $T\subseteq \sembrack{M}$
\item $|T|\le \lfloor\frac{(\lfloor q\rfloor +1)^2}{\lfloor q\rfloor
    +1 -q} \rfloor +1$
\item $\forall M'.T\subseteq\sembrack{M'}\Rightarrow M'\not\in B_{\it
    GE}[U]$
\end{enumerate}
We show that if $M\not\in B_{\it GE}[U]$, then there exists $T$ such
that
\begin{itemize}
\item $T\subseteq \sembrack{M}$
\item $|T|\le \lfloor\frac{(\lfloor q\rfloor +1)^2}{\lfloor q\rfloor
    +1 -q} \rfloor +1$
\item ${\it GE}[U](M')>q$ where $\sembrack{M'}=T$.
\end{itemize}
Note that ${\it GE}[U](M')>q$ and Lemma~\ref{lem:gemono2} imply the
condition 3 above.  Suppose that $M\not\in B_{\it GE}[U]$.  Then, by
Lemma~\ref{lem:gemax3}, there exists $T\subseteq\sembrack{M}$ such
that $|T|\le \lfloor\frac{(\lfloor q\rfloor +1)^2}{\lfloor q\rfloor +1
-q} \rfloor +1$, and ${\it GE}[U](M')>q$ where $\sembrack{M'}=T$.

Next, we prove $B_{\it GE}[U]$ for programs without low-security inputs
is not $k$-safety for any $k \leq \lfloor\frac{(\lfloor q\rfloor
  +1)^2}{\lfloor q\rfloor +1 -q} \rfloor$.  For a contradiction,
suppose $B_{\it GE}[U]$ is a k-safety property.  Let $M$ be a program
such that
\[
\begin{array}{l}
M(h_1)=o, M(h_2)=o, \dots, M(h_i)=o, \\
M(h_{i+1})=o', M(h_{i+2})=o', \dots, M(h_{n})=o'
\end{array}
\]
where $h_1,h_2,\dots h_n$, and $o,o'$ are distinct,
$n=\lfloor\frac{(\lfloor q\rfloor +1)^2}{\lfloor q\rfloor +1 -q}
\rfloor +1$, and $i=\lfloor q\rfloor +1$.  Let
$\mathbb{H}_o=\aset{h\mid o=M(h)}$ and and
$\mathbb{H}_{o'}=\aset{h\mid o'=M(h)}$.  By Lemma~\ref{lem:geu}, we
have
\[
\begin{array}{rcl} 
  {\it
    GE}[U](M)&=&\frac{n}{2}-\frac{1}{2n}(|\mathbb{H}_o|^2+|\mathbb{H}_{o'}|^2)\\
  &=&i-\frac{i^2}{n}\\
  &=&\lfloor q\rfloor +1 -\frac{(\lfloor q\rfloor+1)^2}{\lfloor\frac{(\lfloor q\rfloor +1)^2}{\lfloor q\rfloor +1 -q}
    \rfloor+1}
\end{array}
\]
Let $p=\lfloor q\rfloor +1$.  If $\frac{(\lfloor q\rfloor
  +1)^2}{\lfloor q\rfloor +1 -q}$ is an integer, then we have
\[
\begin{array}{rcl}
  {\it GE}[U](M)&=&
  p -\frac{p^2}{\lfloor\frac{p^2}{p -q}\rfloor+1}\\
  &=&p -\frac{p^2}{\frac{p^2+p -q}{p -q}}\\
  &=&q(\frac{(p-q)^2}{p^2q+pq -q^2}+1)\\
  &>&q
\end{array}
\]
The last line follows from $p^2q+pq -q^2 = p^2q+q(p-q) > 0$.

Otherwise, we have $\lfloor\frac{(\lfloor q\rfloor +1)^2}{\lfloor
  q\rfloor +1 -q}\rfloor+1=\lceil\frac{(\lfloor q\rfloor +1)^2}{\lfloor
  q\rfloor +1 -q}\rceil>\frac{(\lfloor q\rfloor +1)^2}{\lfloor q\rfloor
  +1 -q}$.  And,
\[
\begin{array}{rcl}
  {\it GE}[U](M)&=&
  p -\frac{p^2}{\lceil\frac{p^2}{p -q}\rceil}\\
  &>&p -\frac{p^2}{\frac{p^2}{p -q}}\\
  &=&q
\end{array}
\]
Hence, we have ${\it GE}[U](M)>q$.  Therefore, $M\not\in B_{\it GE}[U]$.  Then,
there exists $T$ such that $|T|\le k$, $T\subseteq\sembrack{M}$, and
for any $M'$ such that $T\subseteq \sembrack{M'}$, $M'\not\in B_{\it
GE}[U]$.  Let $\bar{M}$ be a program such that $\sembrack{\bar{M}}=T$.
Then, by Lemma~\ref{lem:geu} and Lemma~\ref{lem:gemono2}, we have
\[
\begin{array}{rcl}
  {\it GE}[U](\bar{M})&\le&\frac{n-1}{2}-\frac{1}{2(n-1)}(i^2 + (n-1-i)^2)\\
  &=&i - \frac{i^2}{\lfloor \frac{i^2}{i -q}\rfloor}\\
  &\le& i - \frac{i^2}{\frac{i^2}{i -q}}\\
  &=&q
\end{array}
\]
It follows that $\bar{M}\in B_{\it GE}[U]$.  Recall that
$T\subseteq\sembrack{\bar{M}}$.  Therefore, this leads to a
contradiction.

Next, we prove that $B_{\it GE}[U]$ is $2$-safety for any
$q<\frac{1}{2}$.  It suffices to show that ${\it GE}[U](M)\le q$ iff
$M$ is non-interferent, because non-interference is a $2$-safety
property and not a $1$-safety
property~\cite{mclean:sp94,barthe:csfw04,darvas:spc05}.  We prove that
if ${\it GE}[U](M)\le q$ then $M$ is non-interferent.  The other
direction follows from Theorem~\ref{thm:nonint}.  We prove the
contraposition.  Suppose $M$ is interferent.  It must be the case that
there exist $h$ and $h'$ such that $M(h)\not=M(h')$.  Let $o=M(h)$,
and $o'=M(h')$.  Let $M'$ be a program such that
$\sembrack{M'}=\aset{(h,o),(h',o')}$.  Note that we have
$\sembrack{M'}\subseteq\sembrack{M}$.  By Lemma~\ref{lem:gemono2}, we
have
\[
{\it GE}[U](M')=\frac{1}{2}\le {\it GE}[U](M)
\]
It follows that ${\it GE}[U](M)>q$.
\end{proof}

\begin{lemma}
\label{lem:mel}
  Let $M$ be a program that has a low-security input, a high-security
  input, and a low-security output.  Then, we have
\[
{\it ME}[U](M)=\log\frac{|{\mathbb O}_{\mathbb L}|}{|{\mathbb L}|}
\]
where ${\mathbb O}_{\mathbb L}=\aset{(o,\ell)\mid \exists h.
  o=M(h,\ell)}$, and ${\mathbb L}$ is sample space of the low-security
input.
\end{lemma}
\begin{proof}
By the definition of ${\it ME}$, we have
\[
{\it ME}[U](M)=\log\frac{1}{\mathcal{V}[U](H|L)}-\log\frac{1}{\mathcal{V}[U](H|O,L)}
\]
where
\[
\begin{array}{c}
  \mathcal{V}[U](H|L)=\frac{1}{|\mathbb{H}|}\\
  \mathcal{V}[U](H|O,L)=\frac{|\mathbb{O}_\mathbb{L}|}{|\mathbb{H}||\mathbb{L}|}
\end{array}
\]
It follows that 
\[
{\it ME}[U](M)=\log\frac{|\mathbb{O}_\mathbb{L}|}{|\mathbb{L}|}
\]
\end{proof}

\begin{reftheorem}{\ref{thm:menk}}
  Let $q$ be a constant.  (And let $B_{\it ME}[U]$ take programs with
  low security inputs.) Then, $B_{\it ME}[U]$ is not a $k$-safety
  property for any $k > 0$.
\end{reftheorem}
\begin{proof}
  For a contradiction, suppose $B_{\it ME}[U]$ is a k-safety property.
  Let $M$ be a program such that $M\not\in B_{\it ME}[U]$.  Then,
  there exists $T$ such that $|T|\le k$, $T\subseteq\sembrack{M}$, and
  for any $M'$ such that $T\subseteq \sembrack{M'}$, $M'\not\in
  B_{\it ME}[U]$.  Let
  $T=\aset{((h_1,\ell_1),o_1),\dots,((h_i,\ell_i),o_i)}$.  Let
  $\bar{M}$ be the following program.
\[
\begin{array}{l}
  \bar{M}(h_1,\ell_1)=o_1, \bar{M}(h_2,\ell_2)=o_2, \dots, \bar{M}(h_i,\ell_i)=o_i,\\
  \bar{M}(h_{i+1},\ell_{i+1})=o_i, \bar{M}(h_{i+2},\ell_{i+2})=o_i, \dots, \bar{M}(h_{n},\ell_n)=o_i
\end{array}
\]
where $n=|\mathbb{\bar{H}}||\mathbb{\bar{L}}|$, and $\mathbb{\bar{H}},
\mathbb{\bar{L}}$ are the high security inputs and the low security
inputs of $\bar{M}$.  Then, by Lemma~\ref{lem:mel}, we have
\[
\begin{array}{rcl}
{\it ME}[U](\bar{M})&=&\log\frac{|\mathbb{O}_\mathbb{\bar{L}}|}{|\mathbb{\bar{L}}|}\\
&\le& \log \frac{i+|\mathbb{{\bar{L}}}|}{|\mathbb{\bar{L}}|}
\end{array}
\]
Therefore, for any $q > 0$, there exists $\mathbb{\bar{L}}$ such that
${\it ME}[U](\bar{M}) \leq q$ and $T\subseteq\sembrack{\bar{M}}$.
Therefore, this leads to a contradiction.
\end{proof}

\begin{lemma}
\label{lem:geu2}
Let $M$ be a program that has a high-security input with sample space
$\mathbb{H}$, a low-security input with sample space $\mathbb{L}$, and
a low-security output.  Then, we have
\[
{\it
    GE}[U](M)=\frac{|\mathbb{H}|}{2}-\frac{1}{2|\mathbb{H}||\mathbb{L}|}\sum_{o,\ell}
  |\mathbb{H}_{o,\ell}|^2
\] 
where $\mathbb{H}_{o,\ell}=\aset{h\mid o=M(h,\ell)}$.
\end{lemma}
\begin{proof}
By the definition, we have
\[
\begin{array}{rcl}
{\it GE}[U](M)&=&\mathcal{G}[U](H|L)-\mathcal{G}[U](H|O,L)\\
&=&\sum_\ell U(\ell)\sum_h In(\lambda h'.U(h'|\ell),\mathbb{H},h)U(h|\ell)\\
&&\qquad-\sum_{o,\ell} U(o,\ell)\sum_h In(\lambda h'.U(h'|o,\ell),\mathbb{H}_{o,\ell},h) U(h|o,\ell)\\
&=&\frac{|\mathbb{H}|+1}{2}-\sum_{o,\ell} \frac{|\mathbb{H}_{o,\ell}|}{|\mathbb{H}||\mathbb{L}|}\frac{1}{|\mathbb{H}_{o,\ell}|}\frac{1}{2}|\mathbb{H}_{o,\ell}|(|\mathbb{H}_{o,\ell}|+1)\\
&=&\frac{|\mathbb{H}|}{2}-\frac{1}{2|\mathbb{H}||\mathbb{L}|}\sum_{o,\ell} |\mathbb{H}_{o,\ell}|^2 
\end{array}
\]
\end{proof}

\begin{reftheorem}{\ref{thm:genk}}
Let $q$ be a constant.  (And let $B_{\it GE}[U]$ take programs with
low security inputs.) Then, $B_{\it GE}[U]$ is not a $k$-safety
property for any $k > 0$.
\end{reftheorem}
\begin{proof}
  For a contradiction, suppose $B_{\it GE}[U]$ is a k-safety property.
  Let $M$ be a program such that $M\not\in B_{\it GE}[U]$.  Then,
  there exists $T$ such that $|T|\le k$, $T\subseteq\sembrack{M}$, and
  for any $M'$ such that $T\subseteq \sembrack{M'}$, $M'\not\in
  B_{\it GE}[U]$.  Let
  $T=\aset{((h_1,\ell_1),o_1),\dots,((h_i,\ell_i),o_i)}$.  Let
  $\bar{M}$ be the following program.
\[
\begin{array}{l}
  \bar{M}(h_1,\ell_1)=o_1, \bar{M}(h_2,\ell_2)=o_2, \dots, \bar{M}(h_i,\ell_i)=o_i,\\
  \bar{M}(h_{i+1},\ell_{i+1})=o_i, \bar{M}(h_{i+2},\ell_{i+2})=o_i, \dots \bar{M}(h_{mn},\ell_{mn})=o_i
\end{array}
\]
where $n=|\mathbb{\bar{H}}|$ and $m=|\mathbb{\bar{L}}|$, and
$\mathbb{\bar{H}}, \mathbb{\bar{L}}$ are the high security inputs and
the low security inputs of $\bar{M}$.  Then, by Lemma~\ref{lem:geu2},
we have
\[
\begin{array}{rcl}
  {\it GE}[U](\bar{M})&=&\frac{n}{2}-\frac{1}{2mn}\sum_{o,\ell}|\mathbb{H}_{o,\ell}|^2\\
  &\le&\frac{n}{2}-\frac{1}{2mn}(in+(m-i)n^2)\\
  &=&\frac{1}{2mn}(-in+in^2)
\end{array}
\]
Therefore, for any $q > 0$, there exists $\mathbb{\bar{L}}$ such that
${\it GE}[U](\bar{M}) \leq q$ and $T\subseteq\sembrack{\bar{M}}$.
Therefore, this leads to a contradiction.
\end{proof}

\begin{reftheorem}{\ref{thm:senk}}
  Let $q$ be a constant and suppose $B_{\it SE}[U]$ only takes
  programs without low security inputs. Then, $B_{\it SE}[U]$ is not a
  $k$-safety property for any $k > 0$.
\end{reftheorem}
\begin{proof}
  For a contradiction, suppose $B_{\it SE}[U]$ is a k-safety property.
  Let $M$ be a program such that $M\not\in B_{\it SE}[U]$.  Then,
  there exists $T$ such that $|T|\le k$, $T\subseteq\sembrack{M}$, and
  for any $M'$ such that $T\subseteq \sembrack{M'}$, $M'\not\in B_{\it
  SE}[U]$.  Let $T=\aset{(h_1,o_1),\dots,(h_i,o_i)}$.  Let $\bar{M}$
  and $\bar{M'}$ be the following programs.
\[
\begin{array}{l}
  \bar{M}(h_1)=o_1,\bar{M}(h_2)=o_2,\dots,\bar{M}(h_i)=o_i,\bar{M}(h_{i+1})=o,\dots,\bar{M}(h_n)=o\\
  \bar{M'}(h_1)=o_1',\bar{M'}(h_2)=o_2',\dots,\bar{M'}(h_i)=o_i',\bar{M'}(h_{i+1})=o',\dots,\bar{M'}(h_n)=o'
\end{array}
\]
where $h_1$, $h_2$, $\dots$, $h_n$ are distinct, and $o_1'$, $o_2'$,
$\dots$, $o_i'$, and $o'$ are distinct.  Then, we have
\[
\begin{array}{rcl}
{\it SE}[U](\bar{M})&\le& {\it SE}[U](\bar{M'})\\
&=&\frac{i}{n}\log n + \frac{n-i}{n}\log\frac{n}{n-i}\\
&=&\log\frac{n}{n-i}+\frac{i}{n}\log(n-i)
\end{array}
\]
Therefore, for any $q > 0$, there exists $\bar{M}$ such that ${\it
SE}[U](\bar{M}) \leq q$ and $T\subseteq\sembrack{\bar{M}}$.
Therefore, this leads to a contradiction.
\end{proof}

\begin{reftheorem}{\ref{thm:be1nk}}
Let $q$ be a constant, and suppose $B_{\it BE1}[\aseq{U,h}]$ only
takes programs without low security inputs.  Then, $B_{\it
  BE1}[\aseq{U,h}]$ is not a $k$-safety property for any $k > 0$.
\end{reftheorem}
\begin{proof}
For a contradiction, suppose $B_{\it BE1}[\aseq{U,h}]$ is
a $k$-safety property.  Let $M$ be a program such that 
\[
M(h_1)=o,\dots, M(h_m)=o,M(h)=o'
\]
where $m=\lfloor 2^q\rfloor$, and $h,h_1,\dots,h_m$ and $o,o'$ are
distinct.  Then, we have ${\it BE}[\aseq{U,h}](M)=\log(m+1)>\log
2^q=q$.  That is, $(M,q)\not\in B_{\it BE1}[\aseq{U,h}]$.  Then, it
must be the case that there is $T$ such that $|T|\le k$,
$T\subseteq\sembrack{M}$, and for any $\bar{M}$ such that $T\subseteq
\sembrack{\bar{M}}$, $(\bar{M},q)\not\in B_{\it BE1}[\aseq{U,h}]$.
Let $T=\aset{(h_1',o_1'),\dots,(h_i',o_i')}$.  Let $\bar{M}$ be the
following program.
\[
\begin{array}{l}
  \bar{M}(h_1')=o_1',
  \bar{M}(h_2')=o_2', 
  \dots, 
  \bar{M}(h_i')=o_i',\\
  \bar{M}(h_{i+1}')=o',
  \bar{M}(h_{i+2}')=o',
\dots,
  \bar{M}(h_{n}')=o'
\end{array}
\]
where 
\begin{itemize}
\item $h_1'$, $h_2'$, $\dots$, $h_{n}'$ are distinct,
\item $h\in\aset{h_1', \dots, h_n'}$,
\item $\aset{o_1', o_2', \dots, o_i'}=\aset{o,o'}$, and
\item $\bar{M}(h) = o'$.
\end{itemize}
Then, we have 
\[
\begin{array}{rcl}
{\it BE}[\aseq{U,h}](\bar{M})\leq-\log\frac{n-i}{n}
\end{array}
\]
It follows that there exists $n$ such that ${\it
BE}[\aseq{U,h}](\bar{M})\le q$.  This leads to a contradiction.
\end{proof}

\begin{lemma}
\label{lem:bemono}
Let $T$ be a trace such that $T=\aset{((h_1,\ell'),o_1), \dots,
  ((h_i,\ell'),o_i)}$ where $o_1, \dots, o_i$ are distinct.  Let $M$
be the program such that $\sembrack{M} = T$ and $M'$ be a program such
that $\sembrack{M'} \supseteq T$.  Then, we have $\max_{h,\ell} {\it
  BE}[\aseq{U,h,\ell}](M')\ge\max_{h,\ell}{\it
  BE}[\aseq{U,h,\ell}](M)$.
\end{lemma}
\begin{proof}
  By definition, we have 
\[
\begin{array}{rcl}
  \max_{h,\ell} {\it BE}[\aseq{U,h,\ell}](M)&=&\log i\\&&\\

  \max_{h,\ell}{\it BE}[\langle U,h,\ell\rangle](M')&=&\max_{h,\ell}-\log\Sigma_{h_0\in
    \aset{h'\mid M'(h',\ell)=M'(h,\ell)}}U(h_0)\\
&\ge&\max_{h}-\log\Sigma_{h_0\in \aset{h'\mid M'(h',\ell')=M'(h,\ell')}}U(h_0)\\
&=&\log \frac{|\aset{h'\mid
      \exists o. M'(h',\ell')=o}|}{\min_{o}|\aset{h'\mid M'(h',\ell')=o}|}
\end{array}
\]
Therefore, it suffices to show that 
\[|\aset{h'\mid \exists
  o.M'(h',\ell')=o}|\ge i \min_{o}\aset{h'\mid M'(h',\ell')=o}
\]  Then,
\[
\begin{array}{l}
  |\aset{h'\mid \exists o. M'(h',\ell')=o}|- i
  \min_{o}\aset{h'\mid M'(h',\ell')=o}\\
  \qquad\ge (m-i)\min_{o}\aset{h'\mid M'(h',\ell')=o}\\
  \qquad\ge 0
\end{array}
\]where $m=|\aset{o\mid \exists h. M'(h,\ell')=o}|$.
\end{proof}

\begin{reftheorem}{\ref{thm:be2nk1}}
  Let $q$ be a constant.  If $q\ge 1$, then $B_{\it BE2}[U]$ is not a
  $k$-safety property for any $k > 0$ even when $B_{\it BE2}[U]$ only
  takes programs without low security inputs.  Otherwise, $q<1$ and
  $B_{\it BE2}[U]$ is a 2-safety property, but it is not a 1-safety
  property.
\end{reftheorem}
\begin{proof}
First, we show for the case $q \geq 1$, $B_{\it BE2}[U]$ is not a
$k$-safety property for any $k > 0$.  For a contradiction, suppose
$B_{\it BE2}[U]$ is a $k$-safety property.  Let $M$ be the program such
that
\[
M = \aset{h_1 \mapsto o,\dots, h_m \mapsto o,h \mapsto o'}
\]
where $m=\lfloor 2^q\rfloor$.  Then, we have ${\it
  BE}[\aseq{U,h}](M)=\log(m+1)>\log 2^q=q$.  That is, $(M,q)\not\in
B_{\it BE2}[U]$.  Then, it must be the case that there exists $T$ such
that $|T|\le k$, $T\subseteq\sembrack{M}$, and for any $M'$ such that
$T\subseteq \sembrack{M'}$, $(M',q)\not\in B_{\it BE2}[U]$.  Note that
for any $M'$ such that $\sembrack{M'} \subsetneq \sembrack{M}$,
$\forall h.{\it BE}[\aseq{U,h}](M') \leq q$, and therefore, it must be
the case that such $T$ must be equal to $\sembrack{M}$.

Let $\bar{M}$ be the following program.
\[
\begin{array}{l}
  \bar{M}(h_1)=o,
  \bar{M}(h_2)=o,
  \dots,
  \bar{M}(h_m)=o,\\
  \bar{M}(h)=o', \bar{M}(h_{m+1})=o',
  \bar{M}(h_{m+2})=o',
\dots,
  \bar{M}(h_{2m-1})=o'
\end{array}
\]
where $h$, $h_1$, $\dots$, $h_{2m-1}$ are distinct.

Then, we have $|\aset{h'\mid \bar{M}(h')=o}|=|\aset{h'\mid
  \bar{M}(h')=o'}|=m$.  Therefore, for any $h'$,
\[
{\it BE}[\aseq{U,h'}](\bar{M})=-\log\frac{m}{2m}=1\le q
\]
This leads to a contradiction.

Next, we prove that $B_{\it BE2}[U]$ is a 2-safety property for any
$q<1$.  It suffices to show that $\forall h,\ell.{\it
  BE}[\aseq{U,h,\ell}](M)\le q$ iff $M$ is non-interferent, because
non-interference is a $2$-safety property and is not a $1$-safety
property~\cite{mclean:sp94,barthe:csfw04,darvas:spc05}.  We prove that
if $\forall h,\ell.{\it BE}[\aseq{U,h,\ell}](M)\le q$ then $M$ is
non-interferent.  The other direction follows from
Theorem~\ref{thm:nonint}.  We prove the contraposition.  Suppose $M$
is interferent.  It must be the case that there exist $h_0$, $h_1$,
and $\ell'$ such that $M(h_0,\ell')\not=M(h_1,\ell')$.  Let
$o=M(h_0,\ell')$, and $o'=M(h_1,\ell')$.  Let $M'$ be a program such
that $\sembrack{M'}=\aset{((h_0,\ell'),o),((h_1,\ell'),o')}$.  Note
that we have $\sembrack{M'}\subseteq\sembrack{M}$.  By
Lemma~\ref{lem:bemono}, we have
\[
\max_{h,\ell} {\it BE}[\aseq{U,h,\ell}](M')=1\le \max_{h,\ell} {\it BE}[\aseq{U,h,\ell}](M)
\] 
It follows that $\neg(\forall h,\ell.{\it BE}[\aseq{U,h,\ell}] \le
q)$.
\end{proof}

\begin{reftheorem}{\ref{thm:secck}}
Let $q$ be a constant.  Then, $B_{\it SECC}$ is
$\lfloor2^q\rfloor+1$-safety, but it is not $k$-safety for any $k \leq
\lfloor2^q\rfloor$.
\end{reftheorem}
\begin{proof}
  Trivial from Theorem~\ref{thm:cck} and the fact that $B_{\it SECC}$
  is equivalent to $B_{\it CC}$.
\end{proof}

\begin{lemma}
\label{lem:mecclemma}
Let $\mu$ be a distribution.  Then, for any low-security input $\ell$,
we have $m_\ell\max_h\mu(h,\ell)\ge\sum_o\max_h\mu(h,\ell,o)$ where
$m_\ell= |M[\mathbb{H},\ell]|$
\end{lemma}
\begin{proof}
\[
\begin{array}{l}
m_\ell\max_h\mu(h,\ell)-\sum_o\max_h\mu(h,\ell,o)\\
\qquad=\sum_o(\max_h\mu(h,\ell)-\max_h\mu(h,\ell,o))\\
\qquad\ge 0
\end{array}
\]
since we have $\forall o.\max_h\mu(h,\ell)\ge\max_h\mu(h,\ell,o)$.
\end{proof}

\begin{reflemma}{\ref{lem:mecceqcc}}
$\max_\mu {\it ME}[\mu](M)={\it CC}(M)$
\end{reflemma}
\begin{proof}
The statement was proved for programs without low security inputs by
Braun et al.~\cite{Braun:09:MFPS}.  We show that the same result holds
for programs with low security inputs.

Let $\ell'$ be a low-security input such that for any $\ell$,
$m_{\ell'}\ge m_\ell$ where $m_{\ell_0} = |M[\mathbb{H},\ell_0]|$.
Let $\mu'$ be a distribution such that
$\forall h.\mu'(h,\ell')=\frac{1}{n}$ where $n$ is the
number of high-security inputs.  We have ${\it CC}(M)= {\it
  ME}[\mu'](M)=\log m_{\ell'}$.  Therefore, it suffices to show that for
any $\mu$, ${\it ME}[\mu'](M)\ge{\it ME}[\mu](M)$.  By definition,
\[
\begin{array}{l}
{\it ME}[\mu'](M) = \log\frac{\sum_o\max_h \mu'(h,\ell',o)}{\max_h
  \mu'(h,\ell')}\\
{\it ME}[\mu](M) = \log\frac{\sum_\ell\sum_o\max_h
  \mu(h,\ell,o)}{\sum_\ell \max_h\mu(h,\ell)}
\end{array}
\]
Therefore, it suffices to show that
\[
\begin{array}{l}
  (\sum_o\max_h \mu'(h,\ell',o))(\sum_\ell \max_h\mu(h,\ell)) \\
\hspace{8em}-(\max_h\mu'(h,\ell'))(\sum_\ell\sum_o\max_h \mu(h,\ell,o))\ge 0
\end{array}
\]
By Lemma~\ref{lem:mecclemma},
\[
\begin{array}{l}
  (\sum_o\max_h \mu'(h,\ell',o))(\sum_\ell \max_h\mu(h,\ell))\\
\qquad\qquad-(\max_h\mu'(h,\ell'))(\sum_\ell\sum_o\max_h \mu(h,\ell,o))\\
 \qquad = \frac{m_{\ell'}}{n}\sum_\ell \max_h\mu(h,\ell)-\frac{1}{n}(\sum_\ell\sum_o\max_h \mu(h,\ell,o))\\
\qquad \ge\frac{m_{\ell'}}{n}(\sum_\ell \max_h\mu(h,\ell)-\sum_\ell \frac{m_\ell}{m_{\ell'}} \max_h\mu(h,\ell))\\
\qquad  \ge 0
\end{array}
\]
Therefore, we have ${\it ME}[\mu'](M)\ge{\it ME}[\mu](M)$.
\end{proof}

\begin{reftheorem}{\ref{thm:mecck}}
  Let $q$ be a constant.  Then, $B_{\it MECC}$ is
  $\lfloor2^q\rfloor+1$-safety, but it is not $k$-safety for any $k
  \leq \lfloor2^q\rfloor$.
\end{reftheorem}
\begin{proof}
Trivial by Theorem~\ref{thm:cck} and Lemma~\ref{lem:mecceqcc}.
\end{proof}

We define the ``normal form'' of the guessing-entropy-based quantitative
information flow expression.
\begin{definition}[Guessing entropy QIF Normal Form]
Let $M$ be a program without low-security input.  The guessing-entropy
based quantitative information flow ${\it GE}[\mu](M)$ can be written
as the linear expression (over $\mu(h_1),\dots,\mu(h_n)$) $\sum_i
a_i\mu(h_i)$ where $\mu(h_1)\ge\mu(h_2)\ge\dots\ge\mu(h_n)$, and each
$a_i$ is a non-negative integer.  We call this expression $\sum_i
a_i\mu(h_i)$ the {\em normal form} of ${\it GE}[\mu](M)$.
\end{definition}

\begin{lemma}
\label{lem:gecc1}
Let $M$ be a program without low-security input.  Let $\sum_i
a_i\mu(h_i)$ be the normal form of ${\it GE}[\mu](M)$.  Then, for any
$x$ such that $x < |\mathbb{H}|$, we have 
\[
\sum_{i\le x} a_i \le
\frac{1}{2}(x-1)x-\frac{1}{2}(j-2)(j-1)
\]
where $j=|\aset{h\in\aset{h_1,\dots,h_{x+1}}\mid M(h)=M(h_{x+1})}|$.
\end{lemma}
\begin{proof}
  By the definition of guessing-entropy-based quantitative information
  flow, we have 
\[
a_i=i-|\aset{h\in\aset{h_1,\dots,h_i}\mid
    M(h)=M(h_i)}|
\]
Therefore, we have
\[
\begin{array}{l}
  \sum_{i\le x} a_i\\ \qquad=\sum_{i\le x}
  (i-|\aset{h\in\aset{h_1,\dots,h_i}\mid M(h)=M(h_i)}|)\\
  \qquad=\frac{1}{2}x(x+1)-\frac{1}{2}(j-1)j\\
\qquad\qquad-\sum_{i\in\aset{i' \leq x\mid
  M(h_{i'})\not=M(h_{x+1})}}|\aset{h\in\aset{h_1,\dots,h_i}\mid
  M(h)=M(h_i)}|\\ 
\qquad\le\frac{1}{2}(x-1)x-\frac{1}{2}(j-2)(j-1)
\end{array}
\]
where $j=|\aset{h\in\aset{h_1,\dots,h_{x+1}}\mid
  M(h)=M(h_{x+1})}|$
\end{proof}

\begin{lemma}
\label{lem:gecc2}
Let $M$ be a program without low-security input.  Let $\sum_i
a_i\mu(h_i)$ be the normal form of ${\it GE}[\mu](M)$.  Then, for any
$x$ such that $x < |\mathbb{H}|$, we have $\sum_{i\le x} a_i \le x a_{x+1}$.
\end{lemma}
\begin{proof}
  By Lemma~\ref{lem:gecc1}, we have 
\[
\sum_{i\le x} a_i \le
  \frac{1}{2}(x-1)x-\frac{1}{2}(j-2)(j-1)
\]
 where $j=|\aset{h\in\aset{h_1,\dots,h_{x+1}}\mid
 M(h)=M(h_{x+1})}|$, that is, $j=x+1-a_{x+1}$.  Therefore,
 it suffices to show that $\frac{1}{2}(x-1)x-\frac{1}{2}(j-2)(j-1)\le
 x a_{x+1}$.  Then,
\[
x a_{x+1} -\frac{1}{2}(x-1)x+\frac{1}{2}(j-2)(j-1)
=\frac{1}{2}((x+\frac{3-2j}{2})^2-\frac{1}{4})
\]
By elementary numerical analysis, it can be shown that for integers
$x$ and $j$ such that $x+1\ge j$, $\frac{1}{2}((x+ \frac{(3-2j)}{2})^2
- \frac{1}{4})$ attains its minimum value $0$ when $x = j-1$.
Therefore, we have $\sum_{i\le x} a_i \le x a_{x+1}$.
\end{proof}

\begin{lemma}
\label{lem:gecc3}
Let $M$ be a program without low-security input.  Let $\mu$ be a
distribution.  Let $h_1,\dots, h_n$ be such that
$\mu(h_1)=\mu(h_2)=\dots=\mu(h_{i-1})>\mu(h_{i})\ge\dots\ge\mu(h_n)$.
Let $\mu'$ be a distribution such that
$\frac{(i-1)\mu(h_1)+\mu(h_i)}{i}=\mu'(h_1)=\dots=\mu'(h_i)$, and
$\forall x. x>i\Rightarrow \mu'(h_x)=\mu(h_x)$.  Then, we have ${\it
  GE}[\mu](M)\le{\it GE}[\mu'](M)$.
\end{lemma}
\begin{proof}
Let $\sum_j a_j\mu(h_j)$ be the normal form of ${\it
GE}[\mu](M)$.  By the construction of $\mu'$, $\sum_j a_j\mu'(h_j)$ is the
normal form of ${\it GE}[\mu'](M)$.  Therefore,
\[
\begin{array}{l}
  {\it GE}[\mu'](M)-{\it GE}[\mu](M)\\ 
\qquad=\sum_j a_j\mu'(h_j)-\sum_j a_j\mu(h_j)\\
\qquad  =(a_1+\dots+a_{i})\frac{(i-1)\mu(h_1)+\mu(h_i)}{i}-(a_1+\dots+a_{i-1})\mu(h_1)-a_i\mu(h_i)\\
\qquad  =\frac{1}{i}((i-1)a_i-A)(\mu(h_1)-\mu(h_{i}))
\end{array}
\]
where $A=a_1+\dots+a_{i-1}$.  Since we have $(i-1)a_i
-(a_1+\dots+a_{i-1})\ge 0$ by Lemma~\ref{lem:gecc2}, and
$\mu(h_1)-\mu(h_{i})>0$, we have
\[
\frac{1}{i}((i-1)a_i-A)(\mu(h_1)-\mu(h_{i}))\ge 0
\]
Therefore, we have ${\it GE}[\mu'](M)\ge{\it GE}[\mu](M)$.
\end{proof}

\begin{reflemma}{\ref{lem:gecc}}
We have $\max_\mu {\it GE}[\mu](M)=\max_{\ell'} {\it GE}[U\otimes
\dot{\ell'}](M)$ where $U\otimes \dot{\ell'}$ denotes $\lambda
h,\ell.{\sf if}\;\ell=\ell'\;{\sf then}\; U(h)\;{\sf else}\;0$.
\end{reflemma}
\begin{proof}
\[
\begin{array}{rcl}
{\it GE}[\mu](M)&=&\sum_\ell\mu(\ell)\sum_i i\mu(h_i |
\ell)-\sum_\ell\sum_o\mu(\ell,o)\sum_i i\mu(h_i| \ell, o)\\
&=&\sum_\ell\mu(\ell)(\sum_i i\mu(h_i |\ell)-\sum_o\sum_i i\mu(h_i,o|\ell))\\
&=&\sum_\ell\mu(\ell){\it GE}[\lambda h.\mu(h|\ell)](M(\ell))
\end{array}
\]
By Lemma~\ref{lem:gecc3}, we have $\max_\mu {\it
  GE}[\mu](M(\ell))={\it GE}[U](M(\ell))$.  Therefore, we have
$\max_\mu {\it GE}[\mu](M)=(\max_{\ell'} {\it GE}[U\otimes
\dot{\ell'}](M))$.
\end{proof}

\begin{lemma}
\label{lem:geccmono}
  Let $M$ and $M'$ be programs such that
  $\sembrack{M'}=\sembrack{M}\cup\aset{((h',\ell'),o)}$ and
  $(h',\ell')\not\in\dom(\sembrack{M})$.  Then, we have $\max_\ell
  {\it GE}[U\otimes\ell](M)\le\max_\ell {\it GE}[U\otimes\ell](M')$.
\end{lemma}
\begin{proof}
By Lemma~\ref{lem:gemono2}, for any $\ell$, we have ${\it
GE}[U\otimes\ell](M)\le {\it GE}[U\otimes\ell](M')$.  Therefore,
$\max_\ell {\it GE}[U\otimes\ell](M)\le\max_\ell {\it
GE}[U\otimes\ell](M')$.
\end{proof}

\begin{reftheorem}{\ref{thm:gecck}}
  Let $q$ be a constant.  If $q\ge\frac{1}{2}$, then, $B_{\it GECC}$
  is $\lfloor\frac{(\lfloor q\rfloor +1)^2}{\lfloor q\rfloor +1 -q}
  \rfloor +1$-safety, but it is not $k$-safety for any $k \leq
  \lfloor\frac{(\lfloor q\rfloor +1)^2}{\lfloor q\rfloor +1 -q}
  \rfloor$.  Otherwise, $q<\frac{1}{2}$ and $B_{\it GECC}$ is
  $2$-safety, but it is not $1$-safety.
\end{reftheorem}
\begin{proof}
By Lemma~\ref{lem:gecc}, $(M,q) \in B_{\it GECC}$ iff $\max_{\ell'}
{\it GE}[U\otimes \dot{\ell'}](M) \leq q$.\footnote{Therefore, for
  programs without low security inputs, this theorem follows from
  Theorem~\ref{thm:gek}.  But, we show that the theorem holds also for
  programs with low security inputs.}  We prove for the case $q \geq
\frac{1}{2}$ by a ``reduction'' to the result of Theorem~\ref{thm:gek}.
The case for $q < \frac{1}{2}$ follows by essentially the same
argument.

First, we show that $B_{\it GECC}$ is $\lfloor\frac{(\lfloor q\rfloor
  +1)^2}{\lfloor q\rfloor +1 -q} \rfloor +1$-safety in this case.  By
the definition of $k$-safety, for any $M$ such that $M\not\in B_{\it
  GECC}$, there exists $T$ such that
\begin{enumerate}
\item $T\subseteq \sembrack{M}$
\item $|T|\le \lfloor\frac{(\lfloor q\rfloor +1)^2}{\lfloor q\rfloor
    +1 -q} \rfloor +1$
\item $\forall M'.T\subseteq\sembrack{M'}\Rightarrow M'\not\in B_{\it
    GECC}$
\end{enumerate}
Suppose that $M\not\in B_{\it GECC}$.  By Lemma~\ref{lem:gecc}, it
must be the case that there exists $\ell'$ such that $\max_\mu {\it
  GE}[\mu](M)={\it GE}[U](M(\ell'))$.  Then, by
Lemma~\ref{lem:gemax3}, there exists $T\subseteq\sembrack{M(\ell')}$
such that $|T|\le \lfloor\frac{(\lfloor q\rfloor +1)^2}{\lfloor
  q\rfloor +1 -q} \rfloor +1$, and ${\it GE}[U](M')>q$ where
$\sembrack{M'}=T$.  Let $T' = \aset{((h,\ell'),o) \mid (h,o) \in T}$.
Then, we have ${\it GE}[U](M'')>q$ where $\sembrack{M''}=T'$.  Finally,
by Lemma~\ref{lem:geccmono}, we have that for any $M'$ such that
$T'\subseteq\sembrack{M'}$, $M'\not\in B_{\it GECC}$, and so $B_{\it
  GECC}$ is $\lfloor\frac{(\lfloor q\rfloor +1)^2}{\lfloor q\rfloor +1
  -q} \rfloor +1$-safety.

To see that $B_{\it GECC}$ is not $k$-safety for any $k \leq
\lfloor\frac{(\lfloor q\rfloor +1)^2}{\lfloor q\rfloor +1 -q}
\rfloor$, recall Theorem~\ref{thm:gek} that $B_{\it GE}[U]$ is not
$k$-safety for such $k$ (even) for low-security-input-free programs.
Therefore, the result follows by Lemma~\ref{lem:gecc}.
\end{proof}

\begin{reftheorem}{\ref{thm:be3ni}}
$(M,q) \in B_{\it BE1CC}[h,\ell]$ iff $M(\ell)$ is non-interferent.
\end{reftheorem}
\begin{proof}
  We prove that if $\forall\mu.{\it BE}[\aseq{\mu,h,\ell}](M)\le q$
  then $M(\ell)$ is non-interferent.  The other direction follows from
  Theorem~\ref{thm:nonint}.  We prove the contraposition.  Suppose
  $M(\ell)$ is interferent, that is, there exist $h_0$ and $h_1$ such
  that $M(h_0,\ell)\not=M(h_1,\ell)$. If $M(h,\ell)\not=M(h_1,\ell)$,
  then let $\mu'$ be a distribution such that
  $\mu'(h_1)=1-\frac{1}{\lfloor 2^q\rfloor +1}$.  Otherwise, let
  $\mu'$ be a distribution such that $\mu'(h_0)=1-\frac{1}{\lfloor
    2^q\rfloor +1}$.  Then, we have
\[
{\it BE}[\aseq{\mu',h,\ell}](M)  \ge \log(\lfloor 2^q\rfloor +1)  > q
\]
\end{proof}

\begin{reftheorem}{\ref{thm:be4ni}}
$(M,q) \in B_{\it BE2CC}$ iff $M$ is non-interferent.
\end{reftheorem}
\begin{proof}
Straightforward from Theorem~\ref{thm:be3ni} and the fact that a program $M$ is
non-interferent iff for all $\ell$, $M(\ell)$ is non-interferent.
\end{proof}

\paragraph*{\bf Notation}
In the proofs below, for convenience, we sometimes use large letters
$H$, $L$, $O$, etc.~to range over boolean variables as well as generic
random variables.  Also, we assume that variables $H$, $H'$, $H_1$,
etc.~are high security boolean variables and $L$, $L'$, $L_i$, $O$,
$O_1$, $O_i$, etc.~are low security boolean variables.

\paragraph*{\bf Majority SAT}
The following PP-hardness results (Theorems~\ref{thm:ppse},
\ref{thm:ppme}, \ref{thm:ppge}, \ref{thm:ppcc}, \ref{thm:ppbe1},
\ref{thm:ppbe2}, \ref{thm:ppmecc}, and \ref{thm:ppgecc}) are proven by
a reduction from MAJSAT, which is a PP-complete problem.  MAJSAT is
defined as follows.
\[
\textrm{MAJSAT}=\aset{\phi\mid \#SAT(\phi)>2^{n-1}}
\]
where $n$ is the number of variables in the boolean formula $\phi$, and
$\#SAT(\phi)$ is the number of satisfying assignments of $\phi$.

\begin{figure}[t]
\[
\begin{array}{l}
S(\psi)\equiv\\
\ \ {\sf case}\;(H',\psi,\vect{H})\\
\ \ \ \ {\sf when}\;({\sf true},{\sf true},\_)\;{\sf then}\;\vect{O}:=\vect{\sf true};O':={\sf true};O'':={\sf true} \\
\ \ \ \ {\sf when}\;({\sf true},{\sf false},\_)\;{\sf then}\;\vect{O}:={\vect H};O':={\sf true};O'':={\sf false}\\
\ \ \ \ {\sf when}\;({\sf false},\_,\vect{\sf true})\;{\sf then}\;{\vect O}:=\vect{\sf true};O':={\sf false};O'':={\sf false}\\
\ \ \ \ {\sf else}\\
\ \ \ \ \ \ {\sf if}\:H_1\\
\ \ \ \ \ \ \ \ {\sf then}\:\vect{O}:=\vect{\sf true};O':={\sf true};O'':={\sf true}\\
\ \ \ \ \ \ \ \ {\sf else}\hspace{5pt}{\vect{O}:={\vect H};O':={\sf false};O'':={\sf false}}
\end{array}
\]
where $H'$, $\vect{H}= H_1, \dots, H_n$, and $O'$, $O''$, $\vec{O}$ are distinct.
\caption{The Boolean Program for Lemma~\ref{lem:semonotone} and Theorem~\ref{thm:ppse}.}
\label{fig:boolenc}
\end{figure}

\begin{lemma}
\label{lem:semonotone}
Let $\vect H$ and $H'$ be distinct boolean random variables.  Let $n$
and $m$ be any non-negative integers such that $n\le 2^{|{\vect H}|}$
and $m\le 2^{|\vect{H}|}$.  Let $\phi_m$ (resp. $\phi_n$) be a formula
over $\vect H$ having $m$ (resp. $n$) satisfying assignments. Then,
$n\le m$ iff ${\it SE}[U](M_m)\le{\it SE}[U](M_n)$.  where $M_n\equiv
S(\phi_n)$, $M_m\equiv S(\phi_m)$, and $S$ is defined in
Figure~\ref{fig:boolenc}.\footnote{The encoding $S$ is defined so
    that MAJSAT is reduced to a bounding problem with a rational
    upper-bound $q$ in Theorem~\ref{thm:ppse} below.  A simpler
    encoding is possible if we were to do a reduction with a
    non-rational $q$.}
\end{lemma}
\begin{proof}
  First, we explain the construction $S(\psi)$ of
  Figure~\ref{fig:boolenc}.  Here, we use ML-like case statements
  (i.e., earlier cases have the precedence).  It is easy to see that the case
  statements can be written as nested if-then-else statements. Note
  that $\vect{O}=\vect{\sf true}$, $O'={\sf true}$, and $O''={\sf
    true}$ iff either $H'\wedge \psi$, or $H'\wedge H_1$ and at least
  one of $H_2,\dots,H_n$ is ${\sf false}$.  For other inputs,
  $S(\psi)$ returns disjoint outputs.  Therefore, the number of inputs
  $h$ such that $S(\psi)(h) = \vect{\sf true}$ is $\#SAT(\psi) +
  2^{|\vect{H}|-1} - 1$, and for the rest of the $2^{|\vect{H}|+1} -
  (\#SAT(\psi) + 2^{|\vect{H}|-1} - 1)$ inputs, $S(\psi)$ returns
  disjoint outputs different from $\vect{\sf true}$.

Therefore, 
\[
\begin{array}{rcl}
  {\it SE}[U](M_n)&=&\frac{n + 2^{x-1}-1}{2^{x+1}}\log\frac{2^{x+1}}{n + 2^{x-1}-1} + \frac{2^x - n + 2^{x-1}+1}{2^{x+1}}\log 2^{x+1}\\
  {\it SE}[U](M_m)&=&\frac{m + 2^{x-1}-1}{2^{x+1}}\log\frac{2^{x+1}}{m + 2^{x-1}-1} + \frac{2^x - m + 2^{x-1}+1}{2^{x+1}}\log 2^{x+1}
\end{array}
\]
where $x=|{\vect H}|$.
\begin{itemize}
\item $\Rightarrow$

  Suppose $n\le m\le 2^{|\vect H|}$.  Let $x=|{\vect H}|$, and let $p$
  and $q$ be positive real numbers such that
  $p=\frac{n+2^{x-1}-1}{2^{x+1}}$ and $q=\frac{m+2^{x-1}-1}{2^{x+1}}$.
  We have $0\le p\le q\le \frac{1}{2}$.  Therefore,
\[
\begin{array}{l}
  {\it SE}[U](M_n)-{\it SE}[U](M_m)\\
  \qquad=p\log\frac{1}{p} + (1-p)\log 2^{x+1}
- q\log\frac{1}{q} - (1-q)\log 2^{x+1}\\
  \qquad \ge p\log(\frac{q}{p}) +(q-p)\log 2^{x+1}\\
  \qquad\ge 0
\end{array}
\]

\item $\Leftarrow$

  We prove the contraposition.  Suppose $m< n\le 2^{|\vect H|}$.  Let
  $x=|{\vect H}|$, and let $p$ and $q$ be positive real numbers such
  that $p=\frac{n+2^{x-1}-1}{2^{x+1}}$ and
  $q=\frac{m+2^{x-1}-1}{2^{x+1}}$.  We have $0\le q< p\le
  \frac{1}{2}$.  Therefore,
\[
\begin{array}{l}
  {\it SE}[U](M_m)-{\it SE}[U](M_n)\\
\qquad=\log(\frac{1}{q})^q + \log p^p + ((1-q)-(1-p))\log 2^{x+1}\\
\qquad\ge\log(\frac{1}{q})^q + \log p^q + (p-q)\log 2^{x+1}\\
\qquad\ge (p-q)\log 2^{x+1}\\
\qquad> 0
\end{array}
\]
\end{itemize}
\end{proof}

\begin{reftheorem}{\ref{thm:ppse}}
$\text{PP}\subseteq B_{\it SE}[U]$
\end{reftheorem}
\begin{proof}
Let $\phi$ be a boolean formula.  Let $\psi$ be a boolean formula such
that $\#SAT(\psi)=2^{n-1}+1$ where $n$ is the number of variables in
$\phi$.  Let $q$ be the number such that
\[
\begin{array}{rcl}
  q&=&{\it SE}[U](S(\psi))\\
&=&\frac{2^{n-1}+1 + 2^{n-1}-1}{2^{n+1}}\log\frac{2^{n+1}}{2^{n-1}+1 + 2^{n-1}-1} + \frac{2^{n} -(2^{n-1}+1) + 2^{n-1}+1}{2^{n+1}}\log 2^{n+1}\\
 &=&\frac{1}{2} + \frac{n+1}{2}
\end{array}
\]
where $S$ is defined in Figure~\ref{fig:boolenc}.
Then,
\[
\begin{array}{rcl}
  (S(\phi),q)\in B_{\it SE}[U](S(\phi))&\textrm{iff}& {\it
  SE}[U](S(\phi)) \leq {\it SE}[U](S(\psi))\\
 & \textrm{iff}&\#SAT(\phi)\ge\#SAT(\psi)\\
 &\textrm{iff}& \phi\in\textrm{MAJSAT}
\end{array}
\]
by Lemma~\ref{lem:semonotone}.  Therefore, we can decide if
$\phi\in\textrm{MAJSAT}$ by deciding if ${\it SE}[U](S(\phi))\le q$.
Note that the boolean program $S(\phi)$ and $q$ can be
constructed in time polynomial in the size of $\phi$.  Therefore, this
is a reduction from \textrm{MAJSAT} to $B_{\it SE}[U]$.
\end{proof}

\begin{figure}[t]
\[
\begin{array}{l}
T(\phi)=\\
\ \ {\sf if}\;\phi\vee H'\\
\ \ \ \ \ \ {\sf then}\;O_f:={\sf true};\vect{O}:=\vect{{\sf false}}\\
\ \ \ \ \ \ {\sf else}\;O_f:={\sf false};\vect{O}:=\vect{H}
\end{array}
\]
where $\vect{H}$ and $H'$ are distinct, and $O_f$ and $\vect O$ are
distinct.
\caption{The Boolean Program for Lemma~\ref{lem:tme}, Lemma~\ref{lem:memonotone}, and Theorem~\ref{thm:ppme}}
\label{fig:boolenc2}
\end{figure}

\begin{lemma}
\label{lem:tme}
Let $\vect{H}$ and $H'$ be distinct boolean variables.  Let $\phi$ be
a boolean formula.  Then, we have ${\it
  ME}[U](T(\phi))=\log(\#SAT(\neg\phi)+1)$ where $T$ is defined in
Figure~\ref{fig:boolenc2}.
\end{lemma}
\begin{proof}
  It is easy to see that the number of outputs of $T(\phi)$ is equal
  to the number of satisfying assignment to $\neg \phi$ plus $1$.
  Therefore, it follows from Lemma~\ref{lem:mel} that ${\it
    ME}[U](T(\phi)) = \log(\#SAT(\neg\phi)+1)$.
\end{proof}

\begin{lemma}
\label{lem:memonotone}
Let $\vect H$ and $H'$ be distinct boolean random variables.  Let $m$
and $n$ be any non-negative integers such that $m\le 2^{|{\vect H}|}$
and $n\le 2^{|\vect{H}|}$.  Let $\phi_m$ (resp. $\phi_n$) be a formula
over $\vect H$ having $m$ (resp. $n$) satisfying assignments. Then,
$n\le m$ iff ${\it ME}[U](M_m)\le{\it ME}[U](M_n)$.  where $M_n\equiv
T(\phi_n)$, $M_m\equiv T(\phi_m)$, and $T$ is defined in
Figure~\ref{fig:boolenc2}.
\end{lemma}
\begin{proof}
  By Lemma~\ref{lem:ccme}, Lemma~\ref{lem:ccloglow}, and
  Lemma~\ref{lem:tme}, we have
${\it ME}[U](T(\phi_m))\le{\it ME}[U](T(\phi_n)))$
iff
$  \log(2^{|\vect H|}-m+1)\le\log(2^{|\vect H|}-n+1)$
iff
$
n\le m
$.
\end{proof}

\begin{reftheorem}{\ref{thm:ppme}}
$\text{PP}\subseteq B_{\it ME}[U]$
 \end{reftheorem}
\begin{proof}
  Let $\phi$ be a boolean formula.  Let $\psi$ be a boolean formula
  such that $\#SAT(\psi)=2^{n-1}+1$ where $n$ is the number of
  variables in $\phi$.  Let $q$ be the number such that
\[
q={\it ME}[U](T(\psi))
=\log(2^n-(2^{n-1}+1)+1)=n-1
\]
where $T$ is defined in Figure~\ref{fig:boolenc2}.
Then, we have
\[
\begin{array}{rcl}
  {\it ME}[U](T(\phi))\le q
& \textrm{iff}& {\it ME}[U](T(\phi)) \le {\it ME}[U](T(\psi))\\
  &\textrm{iff}& \phi\in\textrm{MAJSAT}
\end{array}
\]
by Lemma~\ref{lem:memonotone}.  Therefore, we can decide if
$\phi\in\textrm{MAJSAT}$ by deciding if ${\it ME}[U](T(\phi))\le q$.
Note that $T(\phi)$ and $q$ can be constructed in time polynomial in
the size of $\phi$.  Therefore, this is a reduction from
\textrm{MAJSAT} to $B_{\it ME}[U]$.
\end{proof}

\begin{definition}
  Let $M$ be a function such that $M:\mathbb{A}\rightarrow
  \mathbb{B}$.  For any $o\in \mathbb{B}$, we write $M^{-1}(o)$ to mean
\[
M^{-1}(o)=\aset{i\in \mathbb{A}\mid o=M(i)}
\]
\end{definition}

\begin{lemma}
\label{lem:gemonotone}
Let $\vect H$ and $H'$ be distinct boolean random variables.  Let $n$
and $m$ be non-negative integers such that $n\le 2^{|{\vect H}|}$ and
$m\le 2^{|\vect{H}|}$.  Let $\phi_m$ (resp. $\phi_n$) be a formula
over $\vect H$ having $m$ (resp. $n$) satisfying assignments. Then,
$m\le n$ iff ${\it GE}[U](M_n)\le{\it GE}[U](M_m)$.  where $M_n\equiv
O:=\phi_n\vee H'$ and $M_m\equiv O:=\phi_m\vee H'$.
\end{lemma}
\begin{proof}
By the definition, 
\[
\begin{array}{rcl}
  {\it GE}[U](M)&=&\mathcal{G}(H)-\mathcal{G}(H|O)\\
  &=&\frac{1}{2}(2^{|H|+1})+\frac{1}{2}-\sum_o\sum_{1\le i\le |H|}i U(h_i,o)\\
  &=&2^{|H|}-\frac{1}{2^{|H|+2}}(|M^{-1}({\sf true})|^2+|M^{-1}({\sf
    false})|^2)
\end{array}
\]
Therefore, we have 
\[
{\it GE}[U](M_n)\le {\it GE}[U](M_m)
\]
iff
\[
|M_m^{-1}({\sf true})|^2+|M_m^{-1}({\sf false})|^2
\le|M_n^{-1}({\sf true})|^2+|M_n^{-1}({\sf false})|^2
\]
iff $m\le n$.
\end{proof}

\begin{reftheorem}{\ref{thm:ppge}}
$\text{PP}\subseteq B_{\it GE}[U]$
\end{reftheorem}
\begin{proof}
Let $\phi$ be a boolean formula.  Let $\psi$ be a boolean formula
such that  $\#SAT(\psi)=2^{n-1}+1$ where $n$ is the number of variables in $\phi$.  Let $q$ be the number such that
\[
\begin{array}{rcl}
  q&=&{\it GE}(O:=\psi\vee H)\\
  &=&\frac{2^{n+1}}{2}-\frac{1}{2^{n+2}}(|M^{-1}({\sf true})|^2+|M^{-1}({\sf
    false})|^2)\\
  &=&2^n-\frac{1}{2^{n+2}}((2^{n-1}+1)^2+(2^{n-1}-1)^2)
\end{array}
\]
where $H$ is a boolean variable that does not appear in $\psi$ and $\phi$.  
Then, we have
\[
\begin{array}{rcl}
{\it GE}[U](O:=\phi\vee H)\le q & \textrm{iff}&
{\it GE}[U](O:=\phi\vee H)\leq {\it GE}[U](O:=\psi\vee H) \\
& \textrm{iff}& {\it GE}[U](O:=\phi\vee H)\le q\\
  &\textrm{iff}& \#SAT(\phi)\ge\#SAT(\psi)\\
  &\textrm{iff}& \phi\in\textrm{MAJSAT}
\end{array}
\]
by Lemma~\ref{lem:gemonotone}.  Therefore, we can decide if
$\phi\in\textrm{MAJSAT}$ by deciding if ${\it GE}[U](O:=\phi\vee H)\le
q$.  Note that $O:=\phi\vee H$ and $q$ can be constructed in time
polynomial in the size of $\phi$.  Therefore, this is a reduction from
\textrm{MAJSAT} to $B_{\it GE}[U]$.
\end{proof}

\begin{reftheorem}{\ref{thm:ppcc}}
  $\text{PP}\subseteq B_{\it CC}$
\end{reftheorem}
\begin{proof}
  Straightforward from Lemma~\ref{lem:ccme} and Theorem~\ref{thm:ppme}.
\end{proof}

\begin{figure}[t]
\[
\begin{array}{l}
V(\psi)\equiv\\
\ \ {\sf case}\;(H',H'',\vect{H})\\
\ \ \ \ {\sf when}\;({\sf true},{\sf true},\_)\;{\sf then}\;\\
\ \ \ \ \ \ \ \ \ttif{\psi}{O:={\sf true}}{O:={\sf false}} \\
\ \ \ \ {\sf when}\;({\sf true},{\sf false},\vect{\sf true})\;{\sf then}\;O:={\sf false}\\
\ \ \ \ {\sf when}\;({\sf true},{\sf false},\_)\;{\sf then}\;\\
\ \ \ \ \ \ \ \ \ttif{H_1}{O:={\sf true}}{O:={\sf false}} \\
\ \ \ \ {\sf else}\;{O:={\sf false}}
\end{array}
\]
where ${\vect H}=H_1, \dots, H_h$ is the vector of variables appearing in 
$\psi$, and $\vect H$, $H'$, and $H''$ are distinct.
\caption{The Boolean Program for Lemma~\ref{lem:be2mono}, Theorem~\ref{thm:ppbe1}, and Theorem~\ref{thm:ppbe2}.}
\label{fig:boolenc3}
\end{figure}

\begin{lemma}
\label{lem:be2mono}
Let $\vect H$, $H'$, and $H''$ be distinct boolean random variables.
Let $n$ and $m$ be any non-negative integers such that $n\le
2^{|{\vect H}|}$ and $m\le 2^{|\vect{H}|}$.  Let $\phi_m$ (resp.
$\phi_n$) be a formula over $\vect H$ having $m$ (resp. $n$)
satisfying assignments. Then, $n\le m$ iff $\max_h {\it
  BE}[\aseq{U,h}](M_m)\le \max_h {\it BE}[\aseq{U,h}](M_n)$, where
$M_n\equiv V(\psi_n)$, $M_m\equiv V(\phi_m)$, and $V$ is defined in
Figure~\ref{fig:boolenc3}.
\footnote{As in Lemma~\ref{lem:semonotone}, the encoding is
    chosen so as to reduce MAJSAT to the bounding problem with a
    rational upper-bound.}
\end{lemma}
\begin{proof}
First, we explain the construction $V(\psi)$ of
Figure~\ref{fig:boolenc3}.  Note that $V(\psi) = {\sf true}$ iff
either $H'\wedge H'' \wedge \psi$, or $H' \wedge \neg H'' \wedge H_1$
and at least one of $H_2,\dots,H_n$ is ${\sf false}$.  Therefore,
there are strictly more inputs $h$ such that $V(\psi)(h) = {\sf
  false}$ than inputs $h$ such that $V(\psi)(h) = {\sf true}$.  Hence,
$\max_h {\it BE}[\aseq{U,h}](V(\psi)) = {\it
  BE}[\aseq{U,h'}](V(\psi))$ where $h'$ is any input such that
$V(\psi)(h') = {\sf true}$.

Now, let $x=|\vect H|$.  Then,
\[
\begin{array}{rcl}
  \max_h {\it BE}[\aseq{U,h}](M_n)&=&\log\frac{2^{x+2}}{n + 2^{x-1}-1}\\
  \max_h {\it BE}[\aseq{U,h}](M_m)&=&\log\frac{2^{x+2}}{m + 2^{x-1}-1}\\
\end{array}
\]
Therefore, $n \leq m$ iff $\max_h {\it BE}[\aseq{U,h}](M_m) \leq \max_h {\it BE}[\aseq{U,h}](M_n)$.
\end{proof}

\begin{reftheorem}{\ref{thm:ppbe1}}
$\text{PP}\subseteq B_{\it BE1}[\aseq{U,h,\ell}]$
\end{reftheorem}
\begin{proof}
Let $\phi$ be a boolean formula.  Let $\psi$ be a boolean formula such that $\#SAT(\psi)=2^{n-1}+1$ where $n$ is the number of variables in $\phi$.  Let $q$ be the number such that
\[
q={\it BE}[\aseq{U,h}](V(\psi))=\log\frac{2^{n+2}}{2^{n-1}+1 + 2^{n-1}-1}=2
\]
where $V$ is defined in Figure~\ref{fig:boolenc3}, $h$ is a high
security input such that $h(H') = {\sf true}$, $h(H'') = {\sf false}$,
$h(H_1) = {\sf true}$, and $h(H_2) = {\sf false}$.  Note that
$V(\psi)(h)= V(\phi)(h) = {\sf true}$.  Then, we have
\[
\begin{array}{rcl}
(V(\phi),q)\in B_{\it BE1}[\aseq{U,h}]& \textrm{iff} &
\max_{h'} {\it BE}[\aseq{U,h'}](V(\phi))\leq q\\
& \textrm{iff}& \max_{h'} {\it BE}[\aseq{U,h'}](V(\phi))\\
&&\qquad\qquad\qquad\leq \max_{h'}{\it BE}[\aseq{U,h'}](V(\psi))\\
  &\textrm{iff}& \#SAT(\phi)\ge\#SAT(\psi)\\
  &\textrm{iff}& \phi\in\textrm{MAJSAT}
\end{array}
\]
by Lemma~\ref{lem:be2mono}, and the fact that $\max_{h'}{\it
  BE}[\aseq{U,h'}](V(\phi))={\it BE}[\aseq{U,h}](V(\phi))$ and
$\max_{h'}{\it BE}[\aseq{U,h'}](V(\psi))={\it
  BE}[\aseq{U,h}](V(\psi))$.  Therefore, we can decide if
$\phi\in\textrm{MAJSAT}$ by deciding if ${\it
  BE}[\aseq{U,h}](V(\phi))\le q$.  Note that $V(\phi)$ and $q$ can be
constructed in time polynomial in the size of $\phi$ (in fact, $q$ is
just the constant $2$).  Therefore, this is a reduction from
\textrm{MAJSAT} to $B_{\it BE1}[\aseq{U,h}]$.
\end{proof}

\begin{reftheorem}{\ref{thm:ppbe2}}
$\text{PP}\subseteq B_{\it BE2}[U]$
\end{reftheorem}
\begin{proof}
Let $\phi$ be a boolean formula.  Let $\psi$ be a boolean formula such that $\#SAT(\psi)=2^{n-1}+1$ where $n$ is the number of variables in $\phi$.  Let $q$ be the number such that
\[
q=\max_h {\it BE}[\aseq{U,h}](V(\psi))=\log\frac{2^{n+2}}{2^{n-1}+1 + 2^{n-1}-1}=2
\]
where $V$ is defined in Figure~\ref{fig:boolenc3}.  We have
\[
\begin{array}{rcl}
(V(\phi),q)\in B_{\it BE2}[U]& \textrm{iff}&
\max_h {\it BE}[\aseq{U,h}](V(\phi))\leq q\\
& \textrm{iff}& \max_h {\it BE}[\aseq{U,h}](V(\phi))\leq \max_h {\it BE}[\aseq{U,h}](V(\psi))\\
  &\textrm{iff}& \#SAT(\phi)\ge\#SAT(\psi)\\
  &\textrm{iff}& \phi\in\textrm{MAJSAT}
\end{array}
\]
by Lemma~\ref{lem:be2mono}.  Therefore, we can decide if
$\phi\in\textrm{MAJSAT}$ by deciding if $\max_h {\it
  BE}[\aseq{U,h}](V(\phi))\le q$.  Note that $V(\phi)$ and $q$ can be
constructed in time polynomial in the size of $\phi$ (in fact, $q$ is
just the constant $2$).  Therefore, this is a reduction from
\textrm{MAJSAT} to $B_{\it BE2}[U]$.
\end{proof}

\begin{reftheorem}{\ref{thm:ppsecc}}
  $\text{PP}\subseteq B_{\it SECC}$
\end{reftheorem}
\begin{proof}
  Trivial from Theorem~\ref{thm:ppcc} and the fact that $B_{\it SECC}$
  is equivalent to $B_{\it CC}$.
\end{proof}

\begin{reftheorem}{\ref{thm:ppmecc}}
  $\text{PP}\subseteq B_{\it MECC}$
\end{reftheorem}
\begin{proof}
Straightforward from Lemma~\ref{lem:mecceqcc} and Theorem~\ref{thm:ppcc}.
\end{proof}

\begin{reftheorem}{\ref{thm:ppgecc}}
  $\text{PP}\subseteq B_{\it GECC}$
\end{reftheorem}
\begin{proof}
Straightforward from Lemma~\ref{lem:gecc} and Theorem~\ref{thm:ppge}.
\end{proof}

We have shown in a previous work~\cite{DBLP:conf/csfw/yasuoka2010} that
checking non-interference for loop-free boolean programs is coNP-complete.
\begin{lemma}
\label{lem:niconp}
Checking non-interference is coNP-complete for loop-free boolean programs.
\end{lemma}

\begin{reftheorem}{\ref{thm:conpbe3}}
$B_{\it BE1CC}[h,\ell]$ is coNP-complete.
\end{reftheorem}
\begin{proof}
Straightforward from Lemma~\ref{lem:niconp} and Theorem~\ref{thm:be3ni}. 
\end{proof}

\begin{reftheorem}{\ref{thm:conpbe4}}
$B_{\it BE2CC}$ is coNP-complete.
\end{reftheorem}
\begin{proof}
Straightforward from Lemma~\ref{lem:niconp} and Theorem~\ref{thm:be4ni}. 
\end{proof}

\end{document}